\documentclass[a4paper]{article}
\usepackage{refcount}
\usepackage{graphicx}
\usepackage{authblk}

\usepackage{amsmath,amssymb,amsthm}
\newtheorem{theorem}{Theorem}[section]
\newtheorem{lemma}[theorem]{Lemma}
\newtheorem{proposition}[theorem]{Proposition}
\newtheorem{claim}[theorem]{Claim}
\usepackage[linesnumbered,ruled,vlined]{algorithm2e}
\newtheorem{corollary}[theorem]{Corollary}

\newtheorem{definition}{Definition}[section]
\newtheorem{example}{Example}[section]

\newcommand{\aP}{{\rm ILS}}

\newcommand{\sign}{{\rm sgn}}

\usepackage{framed,ascmac}
\usepackage{algorithmic}
\usepackage[linesnumbered,ruled]{algorithm2e}
\usepackage{float}
\newcommand{\median}{{M}}
\newcommand{\dist}{{\rm dist}}
\usepackage{version}
\includeversion{fullpaper}
\excludeversion{conference}
\usepackage{CJKutf8}
\usepackage{ifthen}
\newcommand{\COMM}[2]{{
\begin{CJK}{UTF8}{ipxm}
\ifthenelse{\equal{#1}{YK}}{\color{blue}}{
\ifthenelse{\equal{#1}{HS}}{\color{red}}{
\ifthenelse{\equal{#1}{KK}}{\color{green}}}}
[#1: #2]
\end{CJK}
}}

\makeatletter
\def\iddots{\mathinner{\mkern1mu\raise\p@
    \hbox{.}\mkern2mu\raise4\p@\hbox{.}\mkern2mu
    \raise7\p@\vbox{\kern7\p@\hbox{.}}\mkern1mu}}
\makeatother

\usepackage[top=30truemm,bottom=30truemm,left=25truemm,right=25truemm]{geometry}

\title{Trichotomy for the reconfiguration problem of \\integer linear systems} 
\date{}

\author[1]{Kei Kimura}
\author[2]{Akira Suzuki}
\affil[1]{Saitama University, Japan. \texttt{kkimura@mail.saitama-u.ac.jp}}
\affil[2]{Tohoku University, Japan. \texttt{a.suzuki@ecei.tohoku.ac.jp}}
\begin{document}

\maketitle

\begin{abstract}
In this paper, we consider the reconfiguration problem of integer linear systems.
In this problem, we are given an integer linear system $I$ and two feasible solutions $\boldsymbol{s}$ and $\boldsymbol{t}$ of $I$,
and then asked to transform $\boldsymbol{s}$ to $\boldsymbol{t}$ by changing a value of only one variable at a time,
while maintaining a feasible solution of $I$ throughout.
$Z(I)$ for $I$ is the complexity index introduced by Kimura and Makino (Discrete Applied Mathematics 200:67--78, 2016), 
which is defined by the sign pattern of the input matrix.
We analyze the complexity of the reconfiguration problem of integer linear systems based on the complexity index $Z(I)$ of given $I$.
We then show that the problem is 
(i) solvable in constant time if $Z(I)$ is less than one,
(ii) weakly coNP-complete and pseudo-polynomially solvable if $Z(I)$ is exactly one, and
(iii) PSPACE-complete if $Z(I)$ is greater than one.
%Our results show that the complexity of the reconfiguration problem of integer linear systems can be captured by the sign pattern of its input matrix.
Since the complexity indices of Horn and two-variable-par-inequality integer linear systems are at most one, 
our results imply that the reconfiguration of these systems are in coNP and pseudo-polynomially solvable.
Moreover, this is the first result that reveals coNP-completeness for a reconfiguration problem, to the best of our knowledge.
\end{abstract}
\section{Introduction}\label{sec:introduction}
\subsection*{Reconfiguration problem}

In \emph{reconfiguration problem} we are asked to transform the current configuration into a desired one by step-by-step operations.
Formally, in this problem, we are given two feasible solutions of a combinatorial problem, 
then we find the transformation between them, such that all intermediate results are also feasible, 
and each step conforms to an adjacency relation defined on feasible solutions.
% In combinatorial reconfiguration, we are asked to transform the current configuration into a desired one by step-by-step operations.
The reconfiguration problem investigates the properties of solution spaces of combinatorial problems, and 
has a deep relationship to the optimization variants of them.
% The combinatorial reconfiguration framework investigates the properties of solution spaces of combinatorial problems, and 
% has a connection to the optimization variants of these problems.
After Ito et al.~\cite{IDHPSUU11} introduced this reconfiguration framework,
many researchers applied this framework to a variety of combinatorial problems, 
including not only graph problems such as independent set, vertex cover, and coloring, 
but also set cover, knapsack problem, and general integer programming problem.
For recent surveys, see~\cite{van13,Nis17}.
% After Ito et al.~\cite{IDHPSUU11} introduced this framework,
% many researchers applied this framework to a variety of search problems,
% including graph problems such as independent set, vertex cover, and coloring,
% and combinatorial optimization problems such as set cover, knapsack problem, and general integer programming problem.
% See, e.g., surveys~\cite{van13,Nis17}.
The reconfiguration problem has many possible applications, particularly on ongoing services, e.g., maintenance of power stations and computer networks.
% The problem which applied this framework is called \emph{reconfiguration problem} and has possible applications on, e.g., 
% maintenance of power stations and computer networks.

The reconfiguration framework has been also applied to the Boolean satisfiability problem (SAT), which is a central problem in computer science.
In the reconfiguration problem of SAT, we are given a Boolean formula $\varphi$ and two feasible solutions (i.e., satisfying assignments) $\boldsymbol{s}$ and $\boldsymbol{t}$ of $\varphi$,
and asked to transform $\boldsymbol{s}$ to $\boldsymbol{t}$ by changing only one variable from true to false or from false to true at a time,
while maintaining a feasible solution of $\varphi$ throughout.
This problem is shown to have dichotomy property in terms of constraint language~\cite{GKM09}.
Namely, it is shown that the reconfiguration problem of SAT is in P if the constraints satisfy a certain property and otherwise PSPACE-complete,
where the property distinguishing these complexities is later corrected by Schwerdtfeger~\cite{Sch14}.
%Here,  a Boolean relation is tight if it satisfies...
We note that the reconfiguration problem is called $st$-connectivity problem in~\cite{GKM09,Sch14}.

\subsection*{Integer linear systems and complexity index}
In this paper, we focus on the reconfiguration problem for \emph{integer linear system} (ILS for short).
%, which includes various problems such as the SAT problem.
In ILS, we are given a matrix $A \in \mathbb{Q}^{m\times n}$, a vector $\boldsymbol{b} \in \mathbb{Q}^m$,
and a positive integer $d$.
A \emph{feasible solution} is an integer vector $\boldsymbol{x} \in D^n$ satisfying $A\boldsymbol{x} \geq \boldsymbol{b}$, where $D = \{0, 1, \dots, d\}$.
We denote by $I=(A,\boldsymbol{b},d)$ an instance of ILS.
ILS can formulate many combinatorial optimization problems
and is a fundamental problem studied in many fields such as mathematical programming and theoretical computer science.
%Therefore, analyzing the complexity of the reconfiguration problem of ILS gives a unified framework for the complexity of many problems.
The \emph{feasibility} problem of ILS asks if there exists a feasible solution of a given instance of ILS.
The feasibility problem of ILS has been intensively studied, 
especially compared to the reconfiguration counter part.
The feasibility problem is strongly NP-hard in general, but several (semi-)tractable subclasses are known to exist.
For example, the problem can be solved
in polynomial time, if $n$ is bounded by some constant \cite{Len83}, or if $A$ is totally unimodular \cite{HoK56}.
Moreover, it can be solved in pseudo-polynomial time if
(i) $m$ is bounded by some constant \cite{Pap81},
(ii) it is a Horn system (i.e., each row of $A$ contains at most one positive element) \cite{Glo64,MaD02}, or
(iii) it is a two-variable-per-inequality (TVPI) system (i.e., each row of $A$ contains at most two nonzero elements)~\cite{HMN93,BaR01}.
It is also known that the problem is weakly NP-hard,
even if $m$ is bounded by some constant or
the system is Horn and TVPI (also known as monotone quadratic)~\cite{Lag85}.

In this paper, we investigate the reconfiguration problem of ILS through the complexity index for ILS introduced in~\cite{KiM16}. 
The complexity index extends a complexity index for SAT introduced in~\cite{BCH94}, and 
classifies the complexity of the feasibility problem of ILS in terms of the sign structure of the input matrix.
For an ILS $I=(A,\boldsymbol{b},d)$, 
the complexity index $Z(I)$ of $I$ is the optimal value of the following linear programming problem (LP) 
with variables $Z,\alpha_1, \dots, \alpha_n$.
\begin{eqnarray}\label{LP}
\begin{array}{cll}
\rm{minimize}     &Z\\
\rm{subject\ to}     &\displaystyle \sum_{j: \sign(a_{ij})=+}\alpha_{j} + \sum_{j : \sign(a_{ij})=-}(1- \alpha_{j}) \leq Z &(i=1, \dots ,m)\\
&0 \leq \alpha_j \leq 1 &(j=1, \dots ,n),
\end{array}
\end{eqnarray}
	where for a real number $a$, its sign is defined as
	\begin{equation}
	\sign (a)=\left \{
	\begin{array}{ll}
	+ &(a>0)\\
	0 &(a=0)\\
	- &(a<0).
	\end{array}
	\right.
	\end{equation}
Since LP~\eqref{LP} depends {\em only} on the sign pattern of $A$, 
the index captures the sign structure of ILSes.
For $\gamma \ge 0$, we denote by $\aP(\gamma)$ the family of ILSes $I$ with $Z(I) \leq \gamma$.
For the feasibility problem of ILS, the following trichotomy result is shown in~\cite{KiM16}: 
(i) $\aP(\gamma)$ is solvable in linear time for any $\gamma < 1$, 
(ii) $\aP(1)$ is weakly NP-complete and pseudo-polynomially solvable, and 
(iii) $\aP(\gamma)$ is strongly NP-complete for any $\gamma > 1$; see also Table~\ref{table:results-ILS} in the next subsection.
It should be noted that ILS(1) includes Horn and TVPI ILSes, well-studied subclasses of ILSes.
In fact, for a Horn ILS $I$, $(Z,\alpha_1, \dots, \alpha_n) = (1, 1, \dots, 1)$ is a feasible solution to LP~\eqref{LP}, 
and thus the optimal value of LP~\eqref{LP} is at most one.
For a TVPI ILS $I$, $(Z,\alpha_1, \dots, \alpha_n) = (1, 1/2, \dots, 1/2)$ is a feasible solution to LP~\eqref{LP}, 
and thus the optimal value of LP~\eqref{LP} is at most one.
Therefore, these ILSes are included in ILS(1).
Horn and TVPI ILSes arise in, e.g., program verification and scheduling, respectively, and 
many algorithms have been devised to solve the feasibility problems of these subclasses~\cite{BaR01,Glo64,HMN93,MaD02}.
On the other hand, ILS(1) can be decomposed to Horn and TVPI ILSes in a certain way~\cite{KiM16}; see also Section~\ref{sec:Z=1}.
It should be also noted that we can recognize which class a given ILS belongs to in linear time,
without solving LP~\eqref{LP}~\cite{KiM16}.
This is useful for practice, since if we recognized in linear time that the index of a given ILS is, 
say, less than one,
then we could use the linear time algorithm to solve the feasibility problem.

\subsection*{Main results of the paper}

In this paper, we consider the reconfiguration problem of ILS.
Namely, we are given an ILS $I$ and two feasible solutions $\boldsymbol{s}$ and $\boldsymbol{t}$ of $I$, and then asked to transform $\boldsymbol{s}$ to $\boldsymbol{t}$ by changing a value of only one variable at a time, while maintaining a feasible solution of $I$ throughout.
We analyze the complexity of this problem using the complexity index described in the previous subsection and show the following three results: the reconfiguration problem of ILS is 
\begin{description}
\item [(i) always yes if the complexity index is less than one] ~\\
\if0 
We show that any feasible solution can be transformed to any other feasible solution 
using the allowed changes in the reconfiguration problem.
Therefore, any instance of the reconfiguration problem is a yes instance.
\fi 

\item [(ii) weakly coNP-complete and pseudo-polynomially solvable if the complexity index]
\hspace{-0.5mm}\textbf{is}
\vspace{-3mm}
\item [~~~~~exactly one] ~\\
As mentioned in the previous subsection, 
both Horn and TVPI ILSes 
are contained in ILS(1)~\cite{KiM16}.
Therefore, the reconfiguration problem of these ILSes are both in coNP and pseudo-polynomially solvable from this result.
Furthermore, SAT can be formulated as ILS with the constant-size numerical inputs, 
by representing each clause $(\bigvee_{j \in L^+} x_j \vee \bigvee_{j \in L^-} \overline{x}_j)$ as
$\sum_{j \in L^+} x_j + \sum_{j \in L^-} (1-x_j) \ge 1$ and setting $d=1$.
Thus, the reconfiguration problem of SAT with index at most 1 is polynomially solvable
from this result.
\item [(iii) PSPACE-complete if the complexity index is greater than one] ~\\
We show that the reconfiguration problem is PSPACE-complete even for SAT.
Combining this result with result (ii), we obtain a complexity dichotomy for the reconfiguration problem of SAT in terms of the complexity index.
We compare this dichotomy result with the dichotomy result in~\cite{GKM09,Sch14} in the next subsection.
\end{description}
From the above results, we obtain a complexity trichotomy for the reconfiguration problem of ILS; see also Table~\ref{table:results-ILS}.

We also analyze how far two feasible solutions can be,
namely,
%the longest shortest path between two feasible solutions.
the maximum of the minimum number of value changes between two feasible solutions.
This can be cast as the analysis of the diameter of the solution graph of ILS,
where the \emph{solution graph} is defined as follows:
the vertex set is the set of feasible solutions and
two vertices are adjacent if their hamming distance (i.e., the number of components having different values) is one.
%We also analyze how far two feasible solutions can be and obtain that the diameter of the solution graph
%(formally defined in Section~\ref{sec:preliminaries})
%of ILS is ${\rm \Theta}(n)$, ${\rm \Theta}(dn)$, and ${\rm \Omega}(d \cdot 3^{\sqrt{\frac{(\gamma-1) n}{8}}})$ if
We then show that the diameter of the solution graph of ILS is ${\rm \Theta}(n)$, ${\rm \Theta}(dn)$, and ${\rm \Omega}(d \cdot 3^{\sqrt{\frac{(\gamma-1) n}{8}}})$ if
the complexity index $\gamma$ is respectively less than one, equal to one, and greater than one; see Table~\ref{table:results-ILS}.

\begin{table}[htb]
	\caption{Results for integer linear systems (results of this paper are in bold). 
Here, $n$ is the number of variables and $d$ is the upper bound of the values of the variables.}\label{table:results-ILS}
	\centering
	\begin{tabular}{clll}
		\hline
		$\aP(\gamma)$ & feasibility~\cite{KiM16}  & reconfiguration & diameter            \\ \hline
		$\gamma < 1$         & P (linear time)         & $\boldsymbol{\rm P}$ {\bf (always yes)}  & $\boldsymbol{{\rm \Theta}(n)}$    \\
		$\gamma = 1$         &  weakly NP-complete   & {\bf weakly coNP-complete}        & $\boldsymbol{{\rm \Theta}(dn)}$ \\
		& pseudo-polynomially solvable    & {\bf pseudo-polynomially solvable}        & \\
		$\gamma > 1$         & strongly NP-complete              & {\bf PSPACE-complete} & $\boldsymbol{{\rm \Omega}(d \cdot 3^{\sqrt{\frac{(\gamma-1) n}{8}}})}$    \\
		\hline
	\end{tabular}
\end{table}

In Table~\ref{table:results-ILS}, 
a problem is pseudo-polynomially solvable if  
it is solvable in polynomial time in the numeric value of the input.
Moreover, a problem is weakly NP-complete (resp., weakly coNP-complete) 
if it is NP-complete (resp., coNP-complete) in the usual sense, and 
strongly NP-complete if it is NP-complete even when all of its numerical parameters are bounded by a polynomial in the size of the input.

%It is known that both Horn and TVPI systems 
%are contained in ILS(1)~\cite{KiM16}.
%Therefore, the reconfiguration problem of these systems are both in coNP and pseudo-polynomially solvable from our results for the case of $\gamma = 1$.
Our pseudo-polynomial solvability is based on the decomposition of any ILS in ILS(1) into Horn and TVPI ILSes introduced in~\cite{KiM16}.
In fact, we first show that the reconfiguration problems of these two ILSes are pseudo-polynomially solvable.
For Horn ILS, 
this is done by extending the greedy algorithm for Horn SAT~\cite{GKM09},
using the fact that the set of solutions of any Horn ILS is
closed under a minimum operation (see Section~\ref{sec:preliminaries} for details).
On the other hand, for TVPI ILS,
extending the algorithm for 2-SAT in~\cite{GKM09} is not straightforward.
%This is because the property unique majority operation on Boolean domain used
This is because the majority operation on the Boolean domain (i.e., the SAT case) is uniquely determined,
whereas there are many majority operations on non-Boolean domains (i.e., the ILS case), and
properties of the set of solutions depend on a majority operation under which it is closed.
%This is because the structure of the set of solutions is relatively simple for Boolean domain (i.e., 2-SAT case),
%whereas that for non-Boolean domain
%since the set of solutions is not closed under minimum operation for a TVPI ILS in general.
We reveal that the solution sets of any TVPI ILS are closed under a median operation.
%, which might be known for experts.
Using this closedness property, we can induce a partial order on the set of solutions and
devise an algorithm that changes values of variables according to the partial order.
%We note that the majority operation on the Boolean domain (i.e., the SAT case) is uniquely determined,
%whereas there are many choices of majority operations on non-Boolean domains (i.e., the ILS case).
%Therefore,
For our coNP-completeness result,
we use the reduction by Lagarias~\cite{Lag85} that shows the weak NP-hardness of the feasibility problem of monotone quadratic ILSes.

For the case of $\gamma < 1$,
we show that the solution graph of ILS is always connected.
Therefore, any instance of the reconfiguration problem is a yes instance.
We show this by reformulating the structural result for ILS with index less than one in~\cite{KiM16}.

For the case of $\gamma > 1$,
we show that the reconfiguration problem is PSPACE-complete even for the SAT problem.
As mentioned above, ILS can formulate SAT by representing each clause $(\bigvee_{j \in L^+} x_j \vee \bigvee_{j \in L^-} \overline{x}_j)$ as
$\sum_{j \in L^+} x_j + \sum_{j \in L^-} (1-x_j) \ge 1$ and setting $d=1$.
Through this formulation, we can also define complexity index for SAT, and
this index actually coincides with the complexity index for SAT problem introduced by Boros et al.~\cite{BCH94}.
Using the structural expression introduced in~\cite{GKM09}, 
we show that the reconfiguration problem is PSPACE-complete for SAT with index greater than one.

Finally, we obtain some positive results complementing the hardness results for ILS(1).
An integer linear system $I = (A,\boldsymbol{b},d)$ is called \emph{unit} if $A \in \{ 0, \pm 1 \}^{m \times n}$ for positive integers $m$ and $n$.
We show that the reconfiguration problem of
unit ILS(1) is solvable in polynomial time.
%Since unit ILS($\gamma$), which contains SAT($\gamma$), is PSPACE-complete for $\gamma > 1$, we obtain a dichotomy result for unit ILS.
Note that unit ILS($\gamma$) is PSPACE-complete for $\gamma > 1$, since it contains SAT($\gamma$).
Therefore, we obtain a dichotomy result for unit ILS.
Interestingly, the diameter of the solution graph of a system in unit ILS(1) is still ${\rm \Theta}(dn)$ and thus the length of the shortest path between two feasible solutions can be exponential in the input size, 
where we note that $d$ is a part of the input and its input size is $\log d$.
%Therefore, polynomial time solvability of the reconfiguration problem of unit TVPI systems is an interesting result.
%Therefore, we cannot output an actual path in any polynomial time algorithm for UTVPI systems.
Hence, we cannot output an actual path in any polynomial time algorithm for unit ILS(1).
Therefore, we devise an algorithm that repeats a certain sequence of value changes implicitly exponential time for these systems.
We also show that the reconfiguration problem of ILS(1) is solvable in polynomial time if the number of variables is a fixed constant.
%Moreover, if the number of the variable is a fixed constant,
%
\subsection*{Related work}
%For SAT problem, from our results,
%the reconfiguration problem has dichotomy property in terms of the complexity index.
Our results imply that the reconfiguration problem for SAT has a dichotomy property in terms of the complexity index.
Namely, 
the problem is in P if $Z(I) \le 1$ and
otherwise (i.e., if $Z(I) > 1$) PSPACE-complete.
This result is incomparable to the dichotomy result in~\cite{GKM09,Sch14},
that is, the set of instances in the polynomially solvable class in~\cite{GKM09,Sch14} differs from that with $Z(I) \le 1$.
This is because the results in~\cite{GKM09,Sch14} concern the restriction on the constraint language, namely constraints used to build a problem instance.
On the other hand, our result (or the complexity index) focuses on the combination of constraints.
%Therefore, there exists an instance that is not in the polynomially solvable class in~\cite{GKM09,Sch14}
%but in the polynomially solvable class in our result, and vice versa.
For example, consider
any one-inequality ILS $a_1x_1 + a_2x_2 + \dots +a_nx_n \ge b$.
Then it is in general within the PSPACE realm in terms of constraint language.
However, the complexity index is zero for the ILS, implying that
the ILS lies in the P realm in our result.
%On the other hand, consider an affine constraint $x_1 \oplus x_2 \oplus x_3 = 1$ in the field $F_2$ of order two,
On the other hand, consider an affine equation $x_1 \oplus x_2 \oplus x_3 = 1$ in the field $F_2$ of order two, where $\oplus$ is an addition modulo two.
Then the equation lies in the P realm in the result in~\cite{GKM09,Sch14}.
The equation can be formulated as an ILS by
\vspace{-1mm}
\begin{equation}
\left\{
\begin{array}{l}
x_1 + x_2 + x_3 \ge 1\\
x_1 - x_2 - x_3 \ge -1\\
-x_1 + x_2 - x_3 \ge -1\\
-x_1 - x_2 + x_3 \ge -1,
\end{array}
\right.
\end{equation}
and $d=1$, and its complexity index $Z$ is defined by the following LP.
%\vspace{-3mm}
\begin{eqnarray}\label{LP:affine_conf}
\begin{array}{cll}
\rm{minimize}     &Z\\
\rm{subject\ to}     &\alpha_1 + \alpha_2 + \alpha_3 \le Z &\\
&\alpha_1 + (1- \alpha_2) + (1- \alpha_3) \le Z &\\
&(1- \alpha_1) + \alpha_2 + (1- \alpha_3) \le Z &\\
&(1- \alpha_1) + (1- \alpha_2) + \alpha_3 \le Z &\\
&0 \le \alpha_1,\alpha_2,\alpha_3 \le 1. &
\end{array}
\end{eqnarray}
By a simple calculation, one can verify that the optimal value of this LP is $3/2$.
Hence, the ILS lies in the PSPACE realm in our result.
Therefore, the set of instances in the polynomially solvable class in~\cite{GKM09,Sch14} differs from that with $Z(I) \le 1$.
It should be noted that Horn and 2-SAT lie in the P realm in both the results.

In the literature,
reconfiguration problems of many combinatorial problems
% reconfiguration problems of many graph problems and combinatorial optimization problems
are investigated~\cite{van13,Nis17}.
For many polynomially solvable problems such as minimum matroid basis and matching problem,
the corresponding reconfiguration problems are shown to be polynomially solvable~\cite{IDHPSUU11}.
One may wonder if any reconfiguration for NP-hard problems tends to be PSPACE-complete and conversely,
any reconfiguration for the problems in P is shown to be polynomially solvable.
However, there are many exceptions.
For example, the reconfiguration of the three-coloring problem is solvable in polynomial time~\cite{JKKPP16},
while the feasibility problem is known to be NP-complete.
Conversely, the reconfiguration of the shortest path problem is PSPACE-complete~\cite{Bon13}, 
while the feasibility problem is trivially in P.
While there exist results for specific problems,
the general framework is still open that delineates a line between easy and hard reconfiguration problems.
In this paper, we obtain a trichotomy result for ILS, a large class of combinatorial optimization problems, 
and we hope the result sheds some light on the general reconfiguration framework.
%Furthermore, while the coNP-completeness result for deciding the connectivity of the solution graph is often shown,
Furthermore,
%while the coNP-completeness result for deciding the connectivity of the solution graph is often shown,
%the coNP-completeness result for reconfiguration is rarely seen.
we obtain the first coNP-completeness result for reconfiguration problems as far as we know.
%Our result is interesting in this sense, since it 
%since it is an important step .
Thus, our result 
gives a new insight to the complexity of reconfiguration problems.

The complexity of a reconfiguration problem is, so far, strongly tied to the diameter of the solution graph.
%The general thesis is that
Namely, a reconfiguration problem tend to be in P if the diameter of the solution graph is polynomially bounded in the input size,
and PSPACE-complete if the diameter can be exponential.
In fact, polynomially solvable reconfiguration problems with exponential diameters  
have been only known for trivial problems such as Tower of Hanoi (which is always yes), to the best of our knowledge.
%an always-yes reconfiguration problem, which is trivially solvable in polynomial time, such that the diameter of the solution graph is exponential in the input size is known, 
%whereas the problem is rarely seen that is not always-yes and the diameter of the solution graph is exponential in the input size but polynomially solvable in the literature.
%In contrast, 
Our result for unit ILS(1) 
%provides an example that the reconfiguration problem is in P even if
%the diameter of the solution graph can be exponential in the input size and the problem is not always yes.
provides a first example that the reconfiguration problem is in P even if 
the diameter of the solution graph can be exponential in the input size and the reconfiguration problem is not always yes.
This is of independent interest.
\if0
However, we reveal a large class of combinatorial problems for which
the feasibility problem is (weakly) NP-complete but
the reconfiguration problem is (weakly) coNP-complete.
Our result might indicate the possibility
the difference between weak and strong NP-hardness
expands to the difference between weak (co)NP-hardness and PSPACE-hardness in reconfiguration problems.
\fi
\\
\\
\indent The rest of the paper is organized as follows.
Section~\ref{sec:preliminaries} formally defines the reconfiguration problem of ILS and observes useful properties.
Section~\ref{sec:Z=1} and \ref{sec:Z=1-tractable} consider the case where $Z(I) = 1$, which is the most technically involved part.
Sections~\ref{sec:Z<1} and \ref{sec:Z>1} present our results for $Z(I) < 1$ and $Z(I) > 1$, respectively.
Finally, we conclude the paper in Section~\ref{sec:conclusion}.

\section{Preliminaries}\label{sec:preliminaries}
We assume that the reader is familiar with the standard graph theoretic terminology as contained, e.g., in~\cite{BoM08}.
\subsection{Definitions}

We first define the integer linear systems (ILSes).
In an ILS, we are given a matrix $A=(A_{ij}) \in \mathbb{Q}^{m\times n}$, a vector $\boldsymbol{b} \in \mathbb{Q}^m$,
and a positive integer $d$, where $m$ and $n$ denote positive integers and $\mathbb{Q}$ denotes the set of rational numbers.
We denote an ILS by $I=(A,\boldsymbol{b},d)$.
A \emph{feasible solution} of $I$ is an integer vector $\boldsymbol{x} \in D^n$ satisfying $A\boldsymbol{x} \geq \boldsymbol{b}$, where $D = \{0, 1, \dots, d\}$.
The feasibility problem of ILS is the problem of finding a feasible solution of a given ILS.
Note that the bounds on variables, i.e., the domain $D$, allow us to analyze the problem in more details, and also ensure that the solution graph defined below is finite.

For an integer $n$, a subset $R \subseteq D^n$ is called an \emph{{\rm (}$n$-ary{\rm )} relation} on $D$.

\begin{definition}[Solution graph]
	For a relation $R \subseteq D^n$,
	we define the solution graph $G(R) = (V(R), E(R))$ as follows:
	$V(R):= R$ and $E(R):= \{ \{ \boldsymbol{x}, \boldsymbol{y} \} \mid \boldsymbol{x}, \boldsymbol{y} \in V(R), \dist(\boldsymbol{x},\boldsymbol{y}) = 1 \} $, where $\dist(\boldsymbol{x},\boldsymbol{y}):= |\{ j \mid x_j \neq y_j \}|$ is the Hamming distance of $\boldsymbol{x}$ and $\boldsymbol{y}$.
	For an ILS $I=(A,\boldsymbol{b},d)$, 
	we denote by $G(I)$ the solution graph of the set of the feasible solutions of $I$, that is, 
	$G(R(I))$ with $R(I):= \{ \boldsymbol{x} \in D^n \mid A\boldsymbol{x} \geq \boldsymbol{b} \}$.
\end{definition}
We call a path from $\boldsymbol{s}$ to $\boldsymbol{t}$ an {$\boldsymbol{s}$-$\boldsymbol{t}$ path}.
If there exists an $\boldsymbol{s}$-$\boldsymbol{t}$ path, we say $\boldsymbol{t}$ is \emph{reachable} from $\boldsymbol{s}$.
Using this definition, we can treat the	reconfiguration problem of ILS as following:
in the reconfiguration problem of ILS,
we are given an ILS $I$ and two feasible solutions $\boldsymbol{s}$ and $\boldsymbol{t}$ of $I$, 
and then we are asked whether 
%there exists an $\boldsymbol{s}$-$\boldsymbol{t}$ path on $G(I)$ or not.
$\boldsymbol{t}$ is reachable from $\boldsymbol{s}$ or not.

%We in this paper analyze the computational complexity of the reconfiguration problem of integer linear system through a complexity index introduced in~\cite{KiM16}.

\subsection{Basic observations}\label{subsec:basic}
In this subsection, we give three lemmas which play important roles in this paper.

The following lemma is an extension of Lemma 4.1 in~\cite{GKM09} to a multiple-valued version, and used throughout the paper.
We say that an $n$-ary relation $R$ is \emph{closed under a $k$-ary operation $f: D^k \rightarrow D$} if
for every $\boldsymbol{a}^1, \boldsymbol{a}^2, \dots, \boldsymbol{a}^k \in R$,
the tuple $(f(a^1_1, a^2_1, \dots, a^k_1), \dots, f(a^1_n, \dots, a^k_n))$ is in $R$.
We denote this tuple by $f(\boldsymbol{a}^1, \dots, \boldsymbol{a}^k)$.
An operation $f: D^k \rightarrow D$ is called \emph{idempotent} if $f(x, x, \dots, x) = x$ holds for any $x \in D$.

\begin{lemma}\label{lem:connected_component_closedness}
If a relation $R \subseteq D^n$ is closed under an idempotent operation $f: D^k \rightarrow D$,
%then every connected component of $G(R)$ is closed under $f$.
then for each connected component $G'$ in $G(R)$, $V(G')$ is closed under $f$.
\end{lemma}

\begin{proof}
The proof goes along the same line as the proof of Lemma 4.1 in~\cite{GKM09},
since the proof of Lemma 4.1 only uses the idempotency of $f$.
However, we describe the proof for completeness of this paper.

Let $R \subseteq D^n$ be a relation which is closed under an idempotent operation $f: D^k \rightarrow D$.
Consider vectors $\boldsymbol{a}^1, \boldsymbol{a}^2, \dots, \boldsymbol{a}^k \in R$ that all belong to the same connected component of $G(R)$.
In the rest of this proof, we show that $\boldsymbol{a} = f(\boldsymbol{a}^1, \boldsymbol{a}^2, \dots, \boldsymbol{a}^k)$ also belongs to the connected component.
To this end, we show that there exists a path between $\boldsymbol{a}^1$ and $\boldsymbol{a}$ on $G(R)$.

We first prove that for any integer $i$, $1 \le i \le k$, and $\boldsymbol{s}, \boldsymbol{t} \in R$ in the same connected component of $G(R)$,
there exists a path from $f(\boldsymbol{b}^1, \dots , \boldsymbol{b}^{i-1}, \boldsymbol{s}, \boldsymbol{b}^{i+1}, \dots, \boldsymbol{b}^k)$
to $f(\boldsymbol{b}^1, \dots, \boldsymbol{b}^{i-1}, \boldsymbol{t}, \boldsymbol{b}^{i+1}, \dots, \boldsymbol{b}^k)$
for any $\boldsymbol{b}^1, \dots, \boldsymbol{b}^k \in R$.
Let $\boldsymbol{s} = \boldsymbol{s}^0 \rightarrow \boldsymbol{s}^1 \rightarrow \dots \rightarrow \boldsymbol{s}^\ell = \boldsymbol{t}$ be an $\boldsymbol{s}$-$\boldsymbol{t}$ path.
For every $j \in \{0, 1, \dots ,\ell -1 \}$,
the tuples $f(\boldsymbol{b}^1,\dots , \boldsymbol{b}^{i-1}, \boldsymbol{s}^j, \boldsymbol{b}^{i+1},\dots , \boldsymbol{b}^k)$ and $f(\boldsymbol{b}^1,\dots , \boldsymbol{b}^{i-1}, \boldsymbol{s}^{j+1}, \boldsymbol{b}^{i+1},\dots , \boldsymbol{b}^k)$ belong to the same component of $G(R)$, because they
differ in at most one variable (the variable in which $\boldsymbol{s}^{j}$ and $\boldsymbol{s}^{j+1}$ are different).
Thus $f(\boldsymbol{b}^1,\dots , \boldsymbol{b}^{i-1}, \boldsymbol{s}^0, \boldsymbol{b}^{i+1},\dots , \boldsymbol{b}^k)$
and $f(\boldsymbol{b}^1,\dots , \boldsymbol{b}^{i-1}, \boldsymbol{s}^\ell, \boldsymbol{b}^{i+1},\dots , \boldsymbol{b}^k)$ belong to the same component.

%The above fact implies that there are paths from $\boldsymbol{a}^1 = f(\boldsymbol{a}^1, \boldsymbol{a}^1,\dots , \boldsymbol{a}^1)$ to $f(\boldsymbol{a}^1, \boldsymbol{a}^2, \boldsymbol{a}^1,\dots , \boldsymbol{a}^1)$, from $f(\boldsymbol{a}^1, \boldsymbol{a}^2, \boldsymbol{a}^1,\dots , \boldsymbol{a}^1)$ to $f(\boldsymbol{a}^1, \boldsymbol{a}^2, \boldsymbol{a}^3, \boldsymbol{a}^1,\dots , \boldsymbol{a}^1)$ $,\dots,$ and from $f(\boldsymbol{a}^1, \boldsymbol{a}^2,\dots , \boldsymbol{a}^{k-1} , \boldsymbol{a}^1)$ to $f(\boldsymbol{a}^1, \boldsymbol{a}^2,\dots , \boldsymbol{a}^k) = \boldsymbol{a}$.
Therefore, there exist paths from $\boldsymbol{a}^1 = f(\boldsymbol{a}^1, \boldsymbol{a}^1,\dots , \boldsymbol{a}^1)$ to $f(\boldsymbol{a}^1, \boldsymbol{a}^2, \boldsymbol{a}^1,\dots , \boldsymbol{a}^1)$, from $f(\boldsymbol{a}^1, \boldsymbol{a}^2, \boldsymbol{a}^1,\dots , \boldsymbol{a}^1)$ to $f(\boldsymbol{a}^1, \boldsymbol{a}^2, \boldsymbol{a}^3, \boldsymbol{a}^1,\dots , \boldsymbol{a}^1)$ $,\dots,$ and from $f(\boldsymbol{a}^1, \boldsymbol{a}^2,\dots , \boldsymbol{a}^{k-1} , \boldsymbol{a}^1)$ to $f(\boldsymbol{a}^1, \boldsymbol{a}^2,\dots , \boldsymbol{a}^k) = \boldsymbol{a}$.
Thus there exists a path between $\boldsymbol{a}_1$ and $\boldsymbol{a}$ on $G(R)$.
\end{proof}

We next see that we can replace any column of $A$ with its opposite vector without changing the reachability.
Observe first that the feasibility of an integer linear system does not change if
we replace a variable $x_j$ with a new variable $x'_j = d - x_j$.
Namely, the feasibility of $I = (A, \boldsymbol{b}, d)$ is equivalent to that of $I' = (A', \boldsymbol{b} - dA_{.j}, d)$,
where $A'$ is obtained from $A$ by replacing the $j$-th column $A_{.j}$ with $-A_{.j}$.
Moreover, the reconfiguration problem of $I = (A, \boldsymbol{b}, d)$ can be reduced to
that of $I' = (A', \boldsymbol{b} - dA_{.j}, d)$ as the following lemma shows.

\begin{lemma}\label{lem:polarity_change}
	Let $I=(A,\boldsymbol{b},d)$ be an ILS and $\boldsymbol{s},\boldsymbol{t}$ be solutions of $I$.
	%Then, $\boldsymbol{s}$ and $\boldsymbol{t}$ are connected in $G(I)$ if and only if
	Then, for any $j \in \{1, \dots, n\}$,
	$\boldsymbol{t}$ is reachable from $\boldsymbol{s}$ in $G(I)$ if and only if
	%$(s_1, \dots, s_{j-1}, d- s_j, s_{j+1}, \dots, s_n)$ and
	%$(t_1, \dots, t_{j-1}, d- t_j, t_{j+1}, \dots, t_n)$ are connected in $G(I')$,
	$(t_1, \dots, t_{j-1}, d- t_j, t_{j+1}, \dots, t_n)$ is reachable from $(s_1, \dots, s_{j-1}, d- s_j, s_{j+1}, \dots, s_n)$ in $G(I')$,
	where $I'=(A', \boldsymbol{b} - dA_{.j}, d)$ and  $A'$ is obtained from $A$ by replacing the $j$-th column $A_{.j}$ with $-A_{.j}$.
\end{lemma}

\begin{proof}
Assume that $\boldsymbol{t}$ is reachable from $\boldsymbol{s}$ in $G(I)$.
Let an $\boldsymbol{s}$-$\boldsymbol{t}$ path be $\boldsymbol{s} = \boldsymbol{s}^0 \rightarrow \boldsymbol{s}^1
\rightarrow \dots \rightarrow \boldsymbol{s}^\ell = \boldsymbol{t}$.
By replacing $\boldsymbol{s}^k$ with $(s^k_1, \dots, s^k_{j-1}, d- s^k_j, s^k_{j+1}, \dots, s^k_n)$ for each $k = 0,1, \dots, \ell$,
we obtain a path from $\boldsymbol{s}' := (s_1, \dots, s_{j-1}, d- s_j, s_{j+1}, \dots, s_n)$ to
$\boldsymbol{t}' := (t_1, \dots, t_{j-1}, d- t_j, t_{j+1, \dots, t_n})$ in $G(I')$.
Hence, $\boldsymbol{t}'$ is reachable from $\boldsymbol{s}'$ in $G(I')$.
The converse can be proven similarly.
\end{proof}
By inductively applying Lemma~\ref{lem:polarity_change},
we can replace any columns of $A$ by their opposite vectors without changing reachability by changing vector $\boldsymbol{b}$ appropriately.

We also use the following lemma in Sections 3 and 4.
A matrix $A \in \mathbb{Q}^{m \times n}$ is \emph{Horn} if each row of $A$ has at most one positive element.
An ILS is called \emph{Horn} if the input matrix is Horn.
The \emph{minimum operation} is a binary operation that outputs 
the smaller value of the two inputs.
For an ILS, a feasible solution $\boldsymbol{x}^*$ is called a \emph{unique minimal solution} of the ILS
if it satisfies $\boldsymbol{x}^* \leq \boldsymbol{x}$ for all the feasible solutions $\boldsymbol{x}$ of the ILS.
Here, for two vectors $\boldsymbol{x}$ and $\boldsymbol{y}$, $\boldsymbol{x} \geq \boldsymbol{y}$ holds if $x_j \geq y_j$ for all $j$.
%We in this paper analyze the computational complexity of the reconfiguration problem of ILS through a complexity index introduced in~\cite{KiM16}.

\begin{lemma}[E.g., \cite{MaD02}]
\label{lem:Horn-min-closed}
The set of feasible solutions of a Horn ILS is closed under the minimum operation.
Since any nonempty relation on $D$ closed under the minimum operation has a unique minimal solution, 
so does any feasible Horn ILS.
\end{lemma}

\section{The general case of $Z(I) = 1$}
\label{sec:Z=1}
%In this section, we show that the reconfiguration problem of ILS with index 1 can be pseudo-polynomially solvable.
In this section, we show that the reconfiguration problem of $\aP(1)$ is
weakly coNP-complete and pseudo-polynomially solvable.

\subsection{Basic Properties}\label{subsec:basic-propeties}
In this subsection, we summarize useful properties of ILS(1).

\subsubsection{Horn integer linear systems}\label{subsec:Horn-property}
In this subsection, we treat Horn ILS.
Recall that ILS is called \emph{Horn} if each row of the input matrix $A$ has at most one positive element.
%It is well-known that the solution set of a Horn integer linear system is min-closed; see, e.g., \cite{MaD02}.
Let $I$ be an Horn ILS.
From Lemma~\ref{lem:Horn-min-closed} in Subsection~\ref{subsec:basic},
the set of feasible solutions of $I$ is closed under the minimum operation.
Since the minimum operation is idempotent, 
each connected component of $G(I)$ is also closed under the minimum operation 
by Lemma~\ref{lem:connected_component_closedness} in Subsection~\ref{subsec:basic}, 
It follows that there exists a unique minimal solution in each connected component by Lemma~\ref{lem:Horn-min-closed}.
We show that any vertex of $G(I)$ is connected to the unique minimal solution in the same connected component via a monotone path.
Here, a path $\boldsymbol{s}^0 \rightarrow \boldsymbol{s}^1 \rightarrow \dots \rightarrow \boldsymbol{s}^\ell$ is \emph{monotone} if
$\boldsymbol{s}^0 \geq \boldsymbol{s}^1 \geq \dots \geq \boldsymbol{s}^\ell$ holds, where we recall that
for two vectors $\boldsymbol{x}$ and $\boldsymbol{y}$, we have $\boldsymbol{x} \geq \boldsymbol{y}$ if and only if $x_j \geq y_j$ holds for all $j$.

\begin{lemma}\label{lem:Horn-monotone-path}
For a Horn system $I$,
each vertex in $G(I)$ is connected via a monotone path to the unique minimal solution in the same connected component.
\end{lemma}

\begin{proof}
	Let $I$ be a Horn ILS.
	Let $\boldsymbol{s}$ be an arbitrary feasible solution of $I$ and
	let $\boldsymbol{s}_{\min}$ be the unique minimal solution in the same component as $\boldsymbol{s}$ on $G(I)$.
	Since they are in the same connected component,
	there exists an $\boldsymbol{s}$-$\boldsymbol{s}_{\min}$ path on $G(I)$.
	Let such a path be $\boldsymbol{s}=\boldsymbol{s}^0 \rightarrow \boldsymbol{s}^1 \rightarrow \dots \rightarrow \boldsymbol{s}^\ell = \boldsymbol{s}_{\min}$.
	Note that this path may not be a monotone path.
	
	Now we show that we can construct a monotone $\boldsymbol{s}$-$\boldsymbol{s}_{\min}$ path.
	Let $\boldsymbol{u}^0 = \boldsymbol{s}^0$ and $\boldsymbol{u}^k = \min(\boldsymbol{u}^{k-1},\boldsymbol{s}^{k})$ for each $k$, $1 \le k \le \ell$.
	Note that $\boldsymbol{u}^\ell = \boldsymbol{s}_{\min}$ by minimality of $\boldsymbol{s}_{\min}$,
	and $\boldsymbol{u}^{k-1} \ge \min(\boldsymbol{u}^{k-1},\boldsymbol{s}^{k}) = \boldsymbol{u}^{k}$ for each $k$, $1 \le k \le \ell$.
	In the rest of the proof, we show that $\dist( \boldsymbol{u}^{k-1} , \boldsymbol{u}^{k}) \le 1$ for each $k$, $1 \le k \le \ell$.
	Then we immediately have a monotone $\boldsymbol{s}$-$\boldsymbol{s}_{\min}$ path $\boldsymbol{s}=\boldsymbol{u}^0 \rightarrow \boldsymbol{u}^1 \rightarrow \dots \rightarrow \boldsymbol{u}^\ell = \boldsymbol{s}_{\min}$ (if needed, we delete the redundant feasible solutions).
	
Let $\boldsymbol{x}, \boldsymbol{y} \in \mathbb{Z}^n$ be arbitrary vectors with $\dist(\boldsymbol{x}, \boldsymbol{y}) \le 1$.
	We can easily have that $\dist(\boldsymbol{x}, \min(\boldsymbol{x},\boldsymbol{y})) \le 1$.
	Using this property, we show that $\dist( \boldsymbol{u}^{k-1} , \boldsymbol{u}^{k}) \le 1$ in the following.
For $k=1$, we have 
$$\dist( \boldsymbol{u}^{0} , \boldsymbol{u}^{1}) = \dist( \boldsymbol{s}^{0} , \min(\boldsymbol{s}^0 , \boldsymbol{s}^{1})) \le 1,$$ 
since $\dist(\boldsymbol{s}^0,\boldsymbol{s}^{1})=1$.
For $k\ge 2$, we have $\boldsymbol{u}^{k-1}  =  \min(\boldsymbol{u}^{k-2} , \boldsymbol{s}^{k-1})$ and 
	\begin{equation}
	\begin{array}{lll}
	\boldsymbol{u}^{k} & = & \min(\boldsymbol{u}^{k-1} , \boldsymbol{s}^{k})\\
	& = & \min(\min(\boldsymbol{u}^{k-2} , \boldsymbol{s}^{k-1}) , \boldsymbol{s}^{k})\\
	& = & \min(\boldsymbol{u}^{k-2} , \min(\boldsymbol{s}^{k-1} , \boldsymbol{s}^{k})),
	\end{array}
	\end{equation}
	where we use associativity of the minimum operation in the last equality.
	Since $\dist(\boldsymbol{s}^{k-1}, \boldsymbol{s}^{k}) = 1$, we have $\dist(\boldsymbol{s}^{k-1} , \min(\boldsymbol{s}^{k-1} , \boldsymbol{s}^{k})) \le 1$.
	Then $\dist( \boldsymbol{u}^{k-1} , \boldsymbol{u}^k) \le 1$ follows as desired.
	This completes the proof.
\end{proof}

\subsubsection{Two-variable-per-inequality (TVPI) integer linear systems}\label{subsec:TVPI-property}
In this subsection, we treat TVPI ILS,
i.e., ILS where each row of the input matrix has at most two nonzero elements.
We first show that the solution set of a TVPI ILS is closed under a median operation,
where a ternary operation $\median:D^3 \rightarrow D$ is the \emph{median operation} on $D$ if it outputs the middle value of the three inputs.
For example, we have $\median(2,4,3)=3$ and $\median(2,5,2)=2$.
The fact might be already known, however, the authors cannot find it in the literature.

	\begin{proposition}\label{prop:TVPI-median-closed}
		The solution set of a TVPI ILS is closed under a median operation.
	\end{proposition}
	\begin{proof}
		Let $\median$ be the median operation on $D$.
		Note that $\min(x,y) \leq \median(x,y,z) \leq \max(x,y)$ holds for all $x,y,z \in D$.

		For a TVPI inequality $ax_i+bx_j \geq c$,
		let $\boldsymbol{x},\boldsymbol{y},\boldsymbol{z} \in D^n$ be solutions of the inequality,
		i.e., we have $ax_i+bx_j \geq c$, $ay_i+by_j \geq c$, and $az_i+bz_j \geq c$.
		We show that $a \cdot \median(x_i,y_i,z_i)+b \cdot \median(x_j,y_j,z_j) \geq c$ holds,
		which proves the proposition.

		Without loss of generality, we assume that $x_i \leq y_i \leq z_i$.
		Hence, we have $\median(x_i,y_i,z_i) = y_i$.
		Then it suffices to show that $a y_i + b \cdot \median(x_j,y_j,z_j) \geq c$ holds.
		We show this for all the sign patterns of $a$ and $b$.

		\begin{description}
			\item{\textbf{Case 1:} $a \geq 0$ and $b \geq 0$}\\
			In this case, we have
			\begin{equation}
			\begin{array}{lll}
			a y_i + b \cdot \median(x_j,y_j,z_j) &\geq& ay_i + b \cdot \min(x_j,y_j)\\
			& = &ay_i + \min(bx_j,by_j) \\
			&= &\min(ay_i+bx_j,ay_i+by_j)\\
			& \geq& \min(ax_i+bx_j,ay_i+by_j)\\
			& \geq &c.
			\end{array}
			\end{equation}
			\item{\textbf{Case 2:} $a<0$ and $b\geq 0$}\\
			In this case, it follows that
			\begin{equation}
			\begin{array}{lll}
			a y_i + b \cdot \median(x_j,y_j,z_j) &\geq& ay_i + b \cdot \min(y_j,z_j)\\
			& = &ay_i + \min(by_j,bz_j) \\
			&= &\min(ay_i+by_j,ay_i+bz_j)\\
			& \geq& \min(ay_i+by_j,az_i+bz_j)\\
			& \geq &c,
			\end{array}
			\end{equation}
			where the second last inequality follows from $ay_i \geq az_i$.
			\item{\textbf{Case 3:} $a\geq 0$ and $b< 0$}\\
			In this case, we have
			\begin{equation}
			\begin{array}{lll}
			a y_i + b \cdot \median(x_j,y_j,z_j) &\geq&ay_i + b \cdot \max(x_j,y_j)\\
			& = &ay_i + \min(bx_j,by_j) \\
			&= &\min(ay_i+bx_j,ay_i+by_j)\\
			& \geq& \min(ax_i+bx_j,ay_i+by_j)\\
			& \geq &c,
			\end{array}
			\end{equation}
			where the second last inequality follows from $ay_i \geq ax_i$.
			\item{\textbf{Case 4:} $a<0$ and $b<0$}\\
			In this case, it follows that
			\begin{equation}
			\begin{array}{lll}
			a y_i + b \cdot \median(x_j,y_j,z_j) &\geq&ay_i + b \cdot \max(y_j,z_j)\\
			& = &ay_i + \min(by_j,bz_j) \\
			&= &\min(ay_i+by_j,ay_i+bz_j)\\
			& \geq& \min(ay_i+by_j,az_i+bz_j)\\
			& \geq &c,
			\end{array}
			\end{equation}
			where the second last inequality follows from $ay_i \geq az_i$.
		\end{description}
		Therefore, we show that in all cases $a y_i + b \cdot \median(x_j,y_j,z_j) \geq c$ holds.
		This completes the proof.
	\end{proof}

We also use the following observation to show our result.
For $p \in D$,
let $\sqcap_p: D^2 \rightarrow D$ be defined as $\sqcap_p(x,y) = \median(p,x,y)$ for any $x,y \in D$,
where $\median$ is the median operation.
Then $\sqcap_p$ is a semilattice operation, which is shown in the following lemma.
Here, a binary operation $D^2 \rightarrow D$ is \emph{semilattice} if
it is (i) associative, (ii) commutative, and (iii) idempotent.

	\begin{lemma}\label{lem:median-semilattice-ILS}
		%Assume that $R \subseteq D^n$ is closed under an operation $m$ of a median algebra.
		For $p \in D$,
		$\sqcap_p: D^2 \rightarrow D$ defined as above is a semilattice operation.
%		let $\sqcap_p: D^2 \rightarrow D$ be defined as above.
%		Then $\sqcap_p$ is a semilattice operation.
	\end{lemma}
	\begin{proof}
		We show that each axiom of semilattice operations holds for $\sqcap_p$.

		For associativity, we have to show that
		\begin{equation}
		\sqcap_p(x, \sqcap_p(y, z)) = \median(p,x,\median(p,y,z)) = \median(p,\median(p,x,y),z) = \sqcap_p(\sqcap_p(x, y), z).
		\end{equation}
		The middle equality can be shown by checking all the possibility of the magnitude relations on $x,y,z$ and $p$.
		For example, if $x \leq y \leq z \leq p$ holds,
		then we have $\median(p,x,\median(p,y,z)) = \median(p,x,z) = z$ and
		$\median(p,\median(p,x,y),z) = \median(p, y, z) = z$, and thus the equality holds.
		We leave the reader to check the equality for the other possibilities.

		\if0
		\begin{equation}
		\begin{array}{lll}
		\sqcap_p(x, \sqcap_p(y, z)) &=& \median(p,x,\sqcap_p(y, z))\\
		& = &\median(p,x,\median(p,y,z)) \\
		&= &\median(\median(p,y,z),p,x) \\
		&=&\median(\median(z,p,y),p,x) \\
		&= &\median(z,p,\median(y,p,x)) \\
		&= &\median(p,\median(y,p,x),z) \\
		&=& \median(p,\median(p,x,y),z) \\
		&=& \median(p,\sqcap_p(x,y),z) \\
		&=& \sqcap_p(\sqcap_p(x, y), z).
		\end{array}
		\end{equation}
		\fi

		For commutativity,
		it follows that
		$\sqcap_p(x, y) = \median(p,x,y) = \median(p,y,x) = \sqcap_p(y, x)$.
		Finally, for idempotency,
		we have $\sqcap_p(x, x) = \median(p,x,x) = x$.
		Hence, $\sqcap_p$ is a semilattice operation.
	\end{proof}

From Lemma~\ref{lem:median-semilattice-ILS},
we can construct a poset $(D,\leq_p)$ induced by $\sqcap_p$ for $p \in D$, 
where for any $x,y \in D$, $x \leq_p y$ if and only if $\sqcap_p(x, y) = x$.
%Note that minimum is a semilattice with the induced poset $(D, \leq)$.

From Lemma~\ref{lem:connected_component_closedness} in Subsection~\ref{subsec:basic} and Proposition~\ref{prop:TVPI-median-closed},
each connected component of $G(I)$ is closed under the median operation on $D$ for a TVPI system $I$, since median operations are idempotent.
For any feasible solution $\boldsymbol{t}$ to $I$, define $M_{\boldsymbol{t}}: (D^n)^2 \rightarrow D^n$ as $M_{\boldsymbol{t}}(\boldsymbol{x}, \boldsymbol{y}) = M(\boldsymbol{t},\boldsymbol{x},\boldsymbol{y})$ for any $\boldsymbol{x},\boldsymbol{y} \in D^n$.
Then, as in Lemma~\ref{lem:connected_component_closedness}, we can show that each connected component of $G(I)$ is closed under $M_{\boldsymbol{t}}$, i.e., for two feasible solutions $\boldsymbol{x}, \boldsymbol{y}$ to $I$ in the same connected component $M_{\boldsymbol{t}}(\boldsymbol{x}, \boldsymbol{y})$ is also in the component.
%Specifically, for any feasible solutions $\boldsymbol{t},\boldsymbol{x},\boldsymbol{y}$ of $I$, $M(\boldsymbol{t},\boldsymbol{x},\boldsymbol{y})$ is also a feasible solution to $I$.
%$\sqcap_j$ for any solution $\boldsymbol{t}$ of $I$ and $j \in \{1, \dots, n\}$.
For any two vectors $\boldsymbol{x}$ and $\boldsymbol{y}$,
let $\boldsymbol{x} \geq_{\boldsymbol{t}} \boldsymbol{y}$ hold if and only if $x_j \geq_{t_j} y_j$ holds for all $j$.
%where $x \geq_j y$ if and only if $\sqcap_{t_j}(x, y) = x$.
Similar to Lemma~\ref{lem:Horn-monotone-path}, we show that there exists a $\boldsymbol{t}$-monotone path from any vertex of $G(I)$ to the $\boldsymbol{t}$-unique minimal solution in the same connected component,
%Here, a path $\boldsymbol{s}^0 \rightarrow \boldsymbol{s}^1 \rightarrow \dots \rightarrow \boldsymbol{s}^\ell$ is \emph{$\boldsymbol{t}$-monotone} if
%$\boldsymbol{s}^0 \geq_{\boldsymbol{t}} \boldsymbol{s}^1 \geq_{\boldsymbol{t}} \dots \geq_{\boldsymbol{t}} \boldsymbol{s}^\ell$ holds.
where a $\boldsymbol{t}$-monotone path and unique $\boldsymbol{t}$-minimality are
defined analogously to a monotone path and unique minimality, respectively.

\begin{lemma}\label{lem:TVPI-monotone-path}
	For a TVPI system $I$ and a feasible solution $\boldsymbol{t}$ to $I$,
	each vertex is connected to the unique $\boldsymbol{t}$-minimal solution in the same connected component via a $\boldsymbol{t}$-monotone path.
\end{lemma}

\begin{proof}
	The proof is similar to the one for Lemma~\ref{lem:Horn-monotone-path}.

	Let $I$ be a TVPI integer linear system and $\boldsymbol{t}$ be an arbitrary feasible solution to $I$.
	Let $\boldsymbol{s}$ be an arbitrary feasible solution of $I$ and
	let $\boldsymbol{x}^*$ be the unique $\boldsymbol{t}$-minimal solution in the same component as $\boldsymbol{s}$ on $G(I)$.
	Since they are in the same connected component,
	there exists an $\boldsymbol{s}$-$\boldsymbol{x}^*$ path on $G(I)$.
	Let such a path be $\boldsymbol{s}=\boldsymbol{s}^0 \rightarrow \boldsymbol{s}^1 \rightarrow \dots \rightarrow \boldsymbol{s}^\ell = \boldsymbol{x}^*$.
	Note that this path may not be a $\boldsymbol{t}$-monotone path.

	Now we show that we can construct a $\boldsymbol{t}$-monotone $\boldsymbol{s}$-$\boldsymbol{x}^*$ path.
	For each $k$, $0 \le k \le \ell$, we define $\boldsymbol{s}^k_*$ as follows:
	\begin{equation}
	\boldsymbol{s}^k_* = \begin{cases}
	\boldsymbol{s}^0 & k = 0 \\
	\median(\boldsymbol{t},\boldsymbol{s}^{k-1}_*,\boldsymbol{s}^{k}) & k \ge 1.
	\end{cases}
	\end{equation}
	By the same discussion on Lemma~\ref{lem:Horn-monotone-path}, we have a $\boldsymbol{t}$-monotone $\boldsymbol{s}$-$\boldsymbol{x}^*$ path $\boldsymbol{s}=\boldsymbol{s}^0_* \rightarrow \boldsymbol{s}^1_* \rightarrow \dots \rightarrow \boldsymbol{s}^\ell_* = \boldsymbol{x}^*$ (if needed, we delete the redundant feasible solutions).
	This completes the proof.
\end{proof}

\subsubsection{Decomposition of ILS(1)}\label{subsec:ZIisone_general-property}
We here recall that any instance of ILS(1) can be decomposed into Horn and TVPI systems.
All the result in this subsection are from~\cite{KiM16}.
It is known that an instance $I$ of ILS(1) admits a $QH$-partition
%q-horn integer linear systems
%Let $I=(A,\boldsymbol{b},d)$ be an instance of an integer linear system with index 1.
Let $\{ 1, \dots, n \}$ be a variable index set.
A partition $Q \cup H = \{ 1, \dots, n \}$ (and $Q \cap H = \emptyset $) is called a $QH$-\emph{partition} of $\{ 1, \dots, n \}$,
if it satisfies the following three conditions:
%\begin{alphaenumerate}
\begin{description}
	\item[(a)] Each row $i$ of $A$ contains at most two nonzero elements $A_{ij}$ with $j \in Q$.
	%Or equivalently, $|(P_i \cup N_i) \cap Q|\leq 2$ holds for all $i= 1, \dots , m$.
	\item[(b)] Each row $i$ of $A$ contains at most one positive element $A_{ij}$ with $j \in H$.
	%Or equivalently, $|P_i  \cap H|\leq 1$ holds for all $i= 1, \dots , m$.
	\item[(c)] If a row $i$ of $A$ contains  a positive element $A_{ij}$ with $j \in H$, then
	the elements $A_{ik}$ with $k \in Q$ are all zeros.
	%Or equivalently,
	%if $P_i  \cap H \not=\emptyset $ then $(P_i \cup N_i) \cap Q= \emptyset$.
%\end{alphaenumerate}
\end{description}
For a $QH$-partition, let $S$ denote the set of rows $i$ of $A$ such that $A_{ij}=0$ for all  $j \in Q$.
Define $\overline{S}:=\{ 1, \dots, m \} \setminus S$.
For a row and column index sets $T$ and $C$,
let $A[T,C]$ denote the submatrix of $A$ whose row and column sets are $T$ and $C$, respectively.
Moreover, for a vector $\boldsymbol{a} \in \mathbb{Q}^k$ and $T \subseteq \{ 1, \dots k \}$,
let $\boldsymbol{a}_T$ denote the restriction of $\boldsymbol{a}$ to $T$.
Then, we can decompose the integer linear system as follows:
\begin{equation}
\left\{
\begin{array}{l}
A[S,H]\boldsymbol{x}_H\geq \boldsymbol{b}_S\\
A[\overline{S},H]\boldsymbol{x}_H + A[\overline{S},Q]\boldsymbol{x}_Q \geq \boldsymbol{b}_{\overline{S}},
\end{array}
\right.
\end{equation}
where we note that $A[S,Q] = 0$ by the definition of $S$.
Moreover, note that by the condition (b) of $QH$-partition, the system
$A[S,H]\boldsymbol{x}_H\geq \boldsymbol{b}_S$ is Horn, i.e.,
each row of $A[S,H]$ contains at most one positive element.
Similarly, the elements of $A[\overline{S},H]$ are nonpositive and
each row of $A[\overline{S},Q]$ contains at most two nonzero elements, respectively
by conditions (c) and (a) of $QH$-partition.

\subsection{Pseudo-polynomial solvability}
In this subsection, we show the following theorem.

\begin{theorem}\label{thm:Z=1}
The reconfiguration problem of $\aP(1)$ is pseudo-polynomially solvable.
\end{theorem}

To show the theorem, we first consider two subclasses of ILS(1).
Namely, we show that the reconfiguration problems of Horn and two-variable-per-inequality (TVPI) ILS
are pseudo-polynomially solvable in Subsections~\ref{subsec:Horn} and \ref{subsec:TVPI}, respectively.
Then, using these results, we show Theorem~\ref{thm:Z=1} in Subsection~\ref{subsec:ZIisone_general}.

\subsubsection{Horn integer linear systems}\label{subsec:Horn}
In this subsection, we treat Horn ILS,
i.e., ILS where each row of the input matrix has at most one positive element.
%It is well-known that the solution set of a Horn integer linear system is min-closed; see, e.g., \cite{MaD02}.
To show that the reconfiguration problem of Horn ILS is pseudo-polynomially solvable,
we use Lemma~\ref{lem:Horn-monotone-path} in Subsection~\ref{subsec:Horn-property}.

\begin{proposition}\label{lem:Horn-pseudoP}
The reconfiguration problem of Horn ILS is pseudo-polynomially solvable.
\end{proposition}

\begin{proof}
Let $I$ be a Horn ILS and let $\boldsymbol{s}$ and $\boldsymbol{t}$ be feasible solutions of $I$.
Let $\boldsymbol{s}_{\min}$ and $\boldsymbol{t}_{\min}$ be the unique minimal solution in the same component as $\boldsymbol{s}$ and $\boldsymbol{t}$ respectively.
Clearly, $\boldsymbol{s}$ and $\boldsymbol{t}$ are connected if and only if $\boldsymbol{s}_{\min} = \boldsymbol{t}_{\min}$ holds.
From Lemma~\ref{lem:Horn-monotone-path},
$\boldsymbol{s}_{\min}$ (resp., $\boldsymbol{t}_{\min}$) can be obtained by greedily following
a smaller feasible solution from $\boldsymbol{s}$ (resp., $\boldsymbol{t}$).
Note that the length of any monotone path is at most $dn$ since a value of one component decreases in each step.
Therefore, we can obtain $\boldsymbol{s}_{\min}$ (resp., $\boldsymbol{t}_{\min}$) in time polynomial in $n,m$ and $d$, which implies that
the reconfiguration problem of Horn ILS is pseudo-polynomially solvable.
\end{proof}

\begin{proposition}\label{cor:Horn-diameter}
%	The diameter of each component of $G(I)$ is ${\rm O}(dn)$ for any Horn ILS $I$.
	The diameter of each component of $G(I)$ is ${\rm \Theta}(dn)$ for any Horn ILS $I$.
\end{proposition}

\begin{proof}
Let $I$ be a Horn ILS.
From Lemma~\ref{lem:Horn-monotone-path},
any two vertices of $G(I)$ in the same component are connected
by two monotone paths via the unique minimal solution in the component.
Since the length of any monotone path is at most $dn$, the diameter of each component of $G(I)$ is at most $2dn = {\rm O}(dn)$.

%On the other hand, we can construct a family of Horn ILSes $I$ with 
%the diameter of $G(I)$ at least $dn$.
We now show that the diameter of $G(I)$ can be ${\rm \Omega}(dn)$ even for a monotone quadratic ILS,
where an ILS is \emph{monotone quadratic} if each inequality has at most one positive coefficient and at most one negative coefficient.
Note that a monotone quadratic ILS is a Horn (and TVPI) ILS.

\begin{example}\label{ex:Z=1-diameter-lower-bound}
Consider the following monotone quadratic ILS.

	\begin{equation}\label{eq:diameter_lower_bound}
	\left\{
	\begin{array}{ll}
	x_j - x_{j+1} \geq 0 & (j=1, \dots, n-1)\\
	x_{j+1} - x_j \geq -1 & (j=1, \dots, n-1)\\
	\end{array}
	\right.
	\end{equation}
	The diameter of the solution graph of ILS~\eqref{eq:diameter_lower_bound} is ${\rm \Omega}(dn)$.
	Indeed, consider a path from $(0, 0, \dots, 0)$ to $(d, d, \dots, d)$.
Then we can increase a value of a variable at most one in each step, since any two consecutive variables can differ by at most one.
Therefore,	the length of the path is at least $dn$.
Thus, the diameter of the solution graph of system~\eqref{eq:diameter_lower_bound} is ${\rm \Omega}(dn)$.
\if0
Moreover, we have a path
$(0, 0, \dots, 0) \rightarrow (1, 0, \dots, 0) \rightarrow (1, 1, 0, \dots, 0) \rightarrow \dots \rightarrow (1,1, \dots, 1) \rightarrow (2,1,\dots, 1) \rightarrow (2,2,1\dots, 1) \rightarrow \dots \rightarrow (2,2,\dots, 2) \rightarrow \dots \rightarrow (d, d, \dots, d, d-1) \rightarrow (d, d, \dots, d, d)$, whose length is $dn$.
Therefore, the diameter of the graph is $dn$.
\fi
\end{example}

Combining the upper and lower bounds, 
we obtain that the diameter of each component of $G(I)$ is ${\rm \Theta}(dn)$ for any Horn ILS $I$.
\end{proof}

\subsubsection{Two-variable-per-inequality (TVPI) integer linear systems}\label{subsec:TVPI}
In this subsection, we treat TVPI ILS,
i.e., ILS where each row of the input matrix has at most two nonzero elements.
%It is well-known that the solution set of a Horn integer linear system is min-closed; see, e.g., \cite{MaD02}.
To show that the reconfiguration problem of TVPI ILS is pseudo-polynomially solvable,
we use Lemma~\ref{lem:TVPI-monotone-path} in Subsection~\ref{subsec:TVPI-property}.

\begin{proposition}\label{cor:TVPI-pseudoP}
	The reconfiguration problem of TVPI ILS is pseudo-polynomially solvable.
\end{proposition}

\begin{proof}
	The proof goes along the same line as that of Lemma~\ref{lem:Horn-pseudoP}.
	Let $I$ be a TVPI integer linear system and $\boldsymbol{s},\boldsymbol{t}$ be solutions of $I$.
	Let $\boldsymbol{x}^*$ be the unique $\boldsymbol{t}$-minimal solution in the same component as $\boldsymbol{s}$.
	Then $\boldsymbol{s}$ and $\boldsymbol{t}$ are connected if and only if $\boldsymbol{x}^* = \boldsymbol{t}$ holds.
	From Lemma~\ref{lem:TVPI-monotone-path},
	$\boldsymbol{x}^*$ can be obtained by greedily following a smaller feasible solution (in terms of $\boldsymbol{t}$) from $\boldsymbol{s}$.
	Note that the length of any $\boldsymbol{t}$-monotone path is at most $dn$ since a value of one component decreases (in terms of $\boldsymbol{t}$) in each step.
	Therefore, we can obtain $\boldsymbol{x}^*$ in time polynomial in $n,m$ and $d$.
	This implies that the reconfiguration problem is pseudo-polynomially solvable.
\end{proof}

\begin{proposition}\label{cor:TVPI-diameter}
%	The diameter of each component of $G(I)$ is ${\rm O}(dn)$ for TVPI integer linear system $I$.
	The diameter of each component of $G(I)$ is ${\rm \Theta}(dn)$ for any TVPI integer linear system $I$.
\end{proposition}

\begin{proof}
Since any two vertices of $G(I)$ in the same component are connected by a path of length at most $dn$,
the diameter of each component of $G(I)$ is ${\rm O}(dn)$.

Moreover, since ILS~\eqref{eq:diameter_lower_bound} is a TVPI system, 
the diameter of each component of $G(I)$ can be ${\rm \Omega}(dn)$.

Combining the upper and lower bounds, 
we obtain that the diameter of each component of $G(I)$ is ${\rm \Theta}(dn)$ for any TVPI ILS $I$.
\end{proof}

\subsubsection{General case of $Z(I) = 1$}\label{subsec:ZIisone_general}
We now show Theorem~\ref{thm:Z=1}, that is, the reconfiguration problem of ILS(1) is pseudo-polynomially solvable.
We describe our algorithm to solve the reconfiguration problem of ILS(1) in Algorithm~\ref{alg:Z=1}; see Subsection~\ref{subsec:ZIisone_general-property} for notation.
The algorithm first solves the reconfiguration problem on the Horn system $A[S,H]\boldsymbol{x}_H\geq \boldsymbol{b}_S$
and then solves the reconfiguration problem on a certain TVPI system.
%\noindent Step 1. Let $I_H = (A_{S,H}, b_S)$ be a Horn system and check if $\boldsymbol{s}_H$ and $\boldsymbol{t}_H$ are connected in $G(I_H)$ by the algorithm in Subsection~\ref{subsec:Horn}.
%If not, then output NO. Otherwise, go to Step 2.

%\noindent Step 2. $\boldsymbol{s}_H$ and $\boldsymbol{t}_H$ are in the same connected component in $G(I_H)$, and they are connected to a unique minimal element $\boldsymbol{u}_H^*$ of the component in $G(I_H)$.
%Let $I_Q=(A_{\overline{S},Q}, b-A_H\boldsymbol{u}_H^*)$.
%Check if $\boldsymbol{s}_Q$ and $\boldsymbol{t}_Q$ are connected in $G(I_Q)$.
%If not, then output NO. Otherwise output YES.

\begin{algorithm}
	\caption{Solving the reconfiguration problem of $I=(A,\boldsymbol{b},d)$ with $Z(I)=1$}
	\label{alg:Z=1}
	\begin{algorithmic}[1]
		\STATE compute a $QH$-partition of $\{ 1, \dots, n \}$ and let $I_H = (A[S,H], \boldsymbol{b}_{S},d)$
		\IF{$\boldsymbol{s}_H$ and $\boldsymbol{t}_H$ are \emph{not} connected in $G(I_H)$}
		\STATE output ``NO'' and halt
		\ELSE
		\STATE compute the unique minimal solution $\boldsymbol{x}^*_H$ in the same component as
		\STATE $\boldsymbol{s}_H$ and $\boldsymbol{t}_H$ in $G(I_H)$ and let $I_Q=(A[\overline{S},Q], \boldsymbol{b}_{\overline{S}}-A[\overline{S},H]\boldsymbol{x}_H^*,d)$.
		\ENDIF
		\IF{ $\boldsymbol{s}_Q$ and $\boldsymbol{t}_Q$ are connected in $G(I_Q)$}
		\STATE output ``YES'' and halt
		\ELSE
		\STATE output ``NO'' and halt
		\ENDIF
	\end{algorithmic}
\end{algorithm}

\begin{lemma}
	\label{lem:Z=1_analysis}
	Algorithm~\ref{alg:Z=1} solves the reconfiguration problem of ILS(1) in time polynomial in $n$, $m$ and $d$.
\end{lemma}

\begin{proof}
We show that $\boldsymbol{s}$ and $\boldsymbol{t}$ are connected in $G(I)$ if and only if Algorithm~\ref{alg:Z=1} outputs ``YES''.

We first prove the if direction.
Assume that Algorithm~\ref{alg:Z=1} outputs ``YES''.
Then we can construct an $\boldsymbol{s}$-$\boldsymbol{t}$ path
$\boldsymbol{s} = (\boldsymbol{s}_H,\boldsymbol{s}_Q) \rightarrow \dots \rightarrow (\boldsymbol{x}^*_H,\boldsymbol{s}_Q) \rightarrow \dots \rightarrow (\boldsymbol{x}^*_H,\boldsymbol{t}_Q) \rightarrow \dots \rightarrow (\boldsymbol{t}_H,\boldsymbol{t}_Q) = \boldsymbol{t}$ in $G(I)$ using monotone paths from $\boldsymbol{s}_H$ to $\boldsymbol{x}^*_H$ and from $\boldsymbol{x}^*_H$ to $\boldsymbol{t}_H$, and a $\boldsymbol{t}_Q$-monotone path from $\boldsymbol{s}_Q$ to $\boldsymbol{t}_Q$, which exist since the algorithm outputs ``YES'' and from Lemmas~\ref{lem:Horn-monotone-path} and \ref{lem:TVPI-monotone-path}.
Since we use monotone paths, vectors from $(\boldsymbol{s}_H,\boldsymbol{s}_Q) \rightarrow \dots \rightarrow (\boldsymbol{x}^*_H,\boldsymbol{s}_Q)$ and $(\boldsymbol{x}^*_H,\boldsymbol{t}_Q) \rightarrow \dots \rightarrow (\boldsymbol{t}_H,\boldsymbol{t}_Q)$ are all solutions of $I$, since the elements of $A[\overline{S},H]$ are nonpositive.
Indeed, for $(\boldsymbol{x}_H,\boldsymbol{s}_Q)$ in $(\boldsymbol{s}_H,\boldsymbol{s}_Q) \rightarrow \dots \rightarrow (\boldsymbol{x}^*_H,\boldsymbol{s}_Q)$, we have $\boldsymbol{x}_H \leq \boldsymbol{s}_H$ by monotonicity and thus $A[\overline{S},H]\boldsymbol{x}_H \geq  A[\overline{S},H]\boldsymbol{s}_H$ holds since $A[\overline{S},H]$ is a nonpositive matrix.
Therefore, we have
\begin{equation}
A[\overline{S},H]\boldsymbol{x}_H + A[\overline{S},Q]\boldsymbol{s}_Q \geq A[\overline{S},H]\boldsymbol{s}_H + A[\overline{S},Q]\boldsymbol{s}_Q  \geq \boldsymbol{b}_{\overline{S}},
\end{equation}
where the second inequality holds since $\boldsymbol{s}$ is a solution to $I$.
Since $\boldsymbol{x}_H$ is a solution to $I_H$, this implies that $(\boldsymbol{x}_H,\boldsymbol{s}_Q)$ is a solution to $I$.
Similarly, any vector in $(\boldsymbol{x}^*_H,\boldsymbol{t}_Q) \rightarrow \dots \rightarrow (\boldsymbol{t}_H,\boldsymbol{t}_Q)$ is a solution to $I$.
Moreover, vectors in $(\boldsymbol{x}^*_H,\boldsymbol{s}_Q) \rightarrow \dots \rightarrow (\boldsymbol{x}^*_H,\boldsymbol{t}_Q)$ are solutions of $I$ by the definition of $I_Q$ in the algorithm.
Therefore, we obtain an $\boldsymbol{s}$-$\boldsymbol{t}$ path in $G(I)$,
implying that $\boldsymbol{s}$ and $\boldsymbol{t}$ are connected in $G(I)$.

We next prove the only-if direction.
We assume that $\boldsymbol{s}$ and $\boldsymbol{t}$ are connected in $G(I)$.
Let $P: \boldsymbol{s} = \boldsymbol{s}^0 \rightarrow \boldsymbol{s}^1
\rightarrow \dots \rightarrow \boldsymbol{s}^\ell = \boldsymbol{t}$ be an $\boldsymbol{s}$-$\boldsymbol{t}$ path in $G(I)$.
Then clearly the restriction of $P$ to variable indices $H$ is an $\boldsymbol{s}_H$-$\boldsymbol{t}_H$ path in $G(I_H)$.
Moreover, the restriction of $P$ to variable indices $Q$ is an $\boldsymbol{s}_Q$-$\boldsymbol{t}_Q$ path in $G(I_Q)$.
This is because
%$I_Q$ is defined in the way that
%$A[\overline{S},Q]\boldsymbol{x}_Q \ge \boldsymbol{b}_{\overline{S}}-A[\overline{S},H]\boldsymbol{x}_H$ is least restrictive among
%$\boldsymbol{x}_H$ in the same component as $\boldsymbol{s}_H$ (and $\boldsymbol{t}_H$).
%Namely,
for any vector $\boldsymbol{u}$ in $P$
we have
$A[\overline{S},Q]\boldsymbol{u}_Q \ge \boldsymbol{b}_{\overline{S}}-A[\overline{S},H]\boldsymbol{u}_H \ge \boldsymbol{b}_{\overline{S}}-A[\overline{S},H]\boldsymbol{x}^*_H$,
since $\boldsymbol{u}_H \ge \boldsymbol{x}^*_H$ holds by unique minimality of $\boldsymbol{x}^*_H$ and $A[\overline{S},H]$ is a nonpositive matrix.
Therefore, Algorithm~\ref{alg:Z=1} outputs ``YES''.
This completes the proof.
\end{proof}

\begin{proof}[Proof of Theorem~\ref{thm:Z=1}.]
This immediately follows from Lemma~\ref{lem:Z=1_analysis}.
\end{proof}

\begin{corollary}\label{cor:Z=1-diameter}
	The diameter of each component of $G(I)$ is ${\rm O}(dn)$ for instance $I$ of ILS(1).
\end{corollary}

\begin{proof}
	Let $\boldsymbol{s}$ and $\boldsymbol{t}$ be vertices in the same component in $G(I)$.
	Consider the $\boldsymbol{s}$-$\boldsymbol{t}$ path
	$\boldsymbol{s} = (\boldsymbol{s}_H,\boldsymbol{s}_Q) \rightarrow \dots \rightarrow (\boldsymbol{x}^*_H,\boldsymbol{s}_Q) \rightarrow \dots \rightarrow (\boldsymbol{x}^*_H,\boldsymbol{t}_Q) \rightarrow \dots \rightarrow (\boldsymbol{t}_H,\boldsymbol{t}_Q) = \boldsymbol{t}$
	constructed in the proof of Lemma~\ref{lem:Z=1_analysis}.
	Then, this is a path of length ${\rm O}(dn)$ from Propositions~\ref{cor:Horn-diameter} and \ref{cor:TVPI-diameter}.
	Therefore, the diameter of each component of $G(I)$ is ${\rm O}(dn)$.
\end{proof}

\begin{theorem}\label{thm:Z=1-diameter-upper-and-lower-bound}
	The diameter of each component of $G(I)$ is ${\rm \Theta}(dn)$ for instance $I$ of ILS(1).
\end{theorem}
\begin{proof}
This follows from corollary~\ref{cor:Z=1-diameter} and Example~\ref{ex:Z=1-diameter-lower-bound}.
\end{proof}

\subsection{Weak coNP-completeness}\label{subsec:coNP-completeness}

In this subsection, we show the following theorem.
\begin{theorem}\label{thm:Z=1,WcoNP}
The reconfiguration problem of $\aP(1)$ is weakly coNP-complete.
\end{theorem}

We show this by showing that the reconfiguration problem of $\aP(1)$ is in coNP and weakly coNP-hard.

\begin{proposition}
The reconfiguration problem of $\aP(1)$ is in coNP.
\end{proposition}

\begin{proof}
To show that $\aP(1)$ reconfiguration is in coNP,
we show how to find a polynomial certificate for no instances, i.e., 
a polynomial certificate such that $\boldsymbol{s}$ and $\boldsymbol{t}$ are disconnected.

Let $(I,\boldsymbol{s}, \boldsymbol{t})$ be a \emph{no} instance of $\aP(1)$ reconfiguration.
%i.e., $\boldsymbol{t}$ is not reachable from $\boldsymbol{s}$ in $G(I)$.
Let $Q\cup H =V$ be a $QH$-partition of the set of variables.
We define $\boldsymbol{x} \ge_{QH} \boldsymbol{y}$ if
$x_i \ge y_i$ for $i \in H$ and $x_i \ge_{t_i} y_i$ for $i \in Q$.

We first observe that we can assume $\boldsymbol{s} \ge_{QH} \boldsymbol{t}$ without loss of generality.
%Indeed, if $\boldsymbol{s}$ and $\boldsymbol{t}$ are connected, then they are both connected to $\min_{QH}(\boldsymbol{s},\boldsymbol{t})$.
Indeed, since $\boldsymbol{s}$ and $\boldsymbol{t}$ are disconnected,
at least one of $\boldsymbol{s}$ and $\boldsymbol{t}$ is not connected to $\min_{QH}(\boldsymbol{s},\boldsymbol{t})$,
where $\min_{QH}(\boldsymbol{x},\boldsymbol{y})_i$ equals to $\min(x_i,y_i)$ if $i \in H$ and
%$\min_{QH}(\boldsymbol{x},\boldsymbol{y})_i = M(t_i,x_i,y_i)$ if $i \in Q$.
$M(t_i,x_i,y_i)$ if $i \in Q$.
Thus, we may find a certificate to the disconnectivity of $\boldsymbol{s}$ and $\min_{QH}(\boldsymbol{s},\boldsymbol{t})$ without loss of generality.
By resetting $\boldsymbol{t} := \min_{QH}(\boldsymbol{s},\boldsymbol{t})$, we obtain $\boldsymbol{s} \ge_{QH} \boldsymbol{t}$.

Let $\boldsymbol{e_i}$ be the $i$-th unit vector, i.e.,
the $i$-th coordinate of $\boldsymbol{e_i}$ is one and the other coordinates are all zeros.
Let $R$ be the set of solutions of $I$, i.e., $R= \{ \boldsymbol{x} \in D^n \mid  A\boldsymbol{x} \ge \boldsymbol{b} \}$.
Consider a vector $\boldsymbol{w}$ satisfying the following four conditions.\\
$C_0$: $\boldsymbol{w} \in R$, \\
$C_1$: $\boldsymbol{w} \le_{QH} \boldsymbol{s}$, \\
$C_2$: $\boldsymbol{w}$ is locally $QH$-minimal, i.e.,
%$\boldsymbol{w} - \alpha\boldsymbol{e_i} \not \in R$ for any $i=1, \dots, n$ and $\alpha > 0$, and \\
$\boldsymbol{w} - \boldsymbol{e_i} \not \in R$ for any $i \in H$ and $\boldsymbol{w} - \sigma_i\boldsymbol{e_i} \not \in R$ for any $i \in Q \setminus \{ i \mid w_i=t_i \}$,
where $\sigma_i = 1$ if $w_i > t_i$ and $-1$ if $w_i < t_i$, and \\
$C_3$: there exists $i \in \{ 1, \dots, n \}$ such that $w_i >_{QH} t_i$.

The following claim shows that $\boldsymbol{w}$ satisfying the four conditions above
is a certificate that $\boldsymbol{s}$ and $\boldsymbol{t}$ are disconnected.

\begin{claim}\label{cl:certificate}
If we have a $\boldsymbol{w}$ satisfying the four conditions $C_i$ $(i=0,1,2,3)$,
then $\boldsymbol{s}$ and $\boldsymbol{t}$ are disconnected.
\end{claim}

\begin{proof}
We first observe that the unique $QH$-minimal element $\boldsymbol{s}_{\min}$ in the same component as $\boldsymbol{s}$ satisfies $\boldsymbol{s}_{\min} \ge_{QH} \boldsymbol{w}$.
Indeed, let
$\boldsymbol{s}=\boldsymbol{s}^0 \rightarrow \boldsymbol{s}^1 \rightarrow \dots \rightarrow \boldsymbol{s}^\ell = \boldsymbol{s}_{\min}$ be a
$QH$-monotone path, which exists by the same argument as Lemmas~\ref{lem:Horn-monotone-path} and \ref{lem:TVPI-monotone-path}.
Consider a path
$\min_{QH}(\boldsymbol{s},\boldsymbol{w}) \rightarrow \min_{QH}(\boldsymbol{s}^1,\boldsymbol{w}) \rightarrow \dots \rightarrow \min_{QH}(\boldsymbol{s}_{\min},\boldsymbol{w})$.
Since $\boldsymbol{w} \in R$ from condition $C_0$ and the solution set of $I$ is $\min_{QH}$-closed, this is a path in $R$.
Moreover, from condition $C_1$, we have $\min_{QH}(\boldsymbol{s},\boldsymbol{w}) = \boldsymbol{w}$ and thus
this is a path from $\boldsymbol{w}$ to $\min_{QH}(\boldsymbol{s}_{\min},\boldsymbol{w})$.
Since $\boldsymbol{w}$ is locally minimal by $C_2$,
we have $\min_{QH}(\boldsymbol{s}_{\min},\boldsymbol{w}) = \boldsymbol{w}$.
Therefore, $\boldsymbol{w} \le \boldsymbol{s}_{\min}$ holds.

On the other hand,
the unique $QH$-minimal element $\boldsymbol{t}_{\rm min}$ in the same component as $\boldsymbol{t}$
does not satisfy $\boldsymbol{t}_{\min} \ge_{QH} \boldsymbol{w}$.
This is because $\boldsymbol{t} \ge_{QH} \boldsymbol{t}_{\min}$ and the condition $C_3$ of $\boldsymbol{w}$.

Therefore, $\boldsymbol{s}_{\rm min} \neq \boldsymbol{t}_{\rm min}$ holds, implying that
$\boldsymbol{s}$ and $\boldsymbol{t}$ belong to different components in $R$.
Thus, they are disconnected.
\end{proof}

We next show that
if $\boldsymbol{s}$ and $\boldsymbol{t}$ are disconnected,
then we always have a certificate satisfying conditions $C_i$ $(i=0,1,2,3)$.

\begin{claim}\label{cl:exists-certificate}
If $\boldsymbol{s}$ and $\boldsymbol{t}$ are disconnected,
then there exists a vector $\boldsymbol{w}$ satisfying conditions $C_i$ $(i=0,1,2,3)$.
\end{claim}
\begin{proof}
We show that the unique $QH$-minimal element $\boldsymbol{s}_{\min}$ in the same component as $\boldsymbol{s}$ satisfies conditions $C_i$ $(i=0,1,2,3)$.
Since $\boldsymbol{s}_{\min} \in R$, Condition $C_0$ is clearly satisfied.
Moreover, since $\boldsymbol{s}_{\min}$ is the unique $QH$-minimal element in the same component $\boldsymbol{s}$,
Condition $C_1$ holds.
Condition $C_2$ follows since $\boldsymbol{s}_{\min}$ is the unique $QH$-minimal element in a component.

For Condition $C_3$, assume otherwise that $\boldsymbol{s}_{\min} \le_{QH} \boldsymbol{t}$ holds.
Consider a path from $\boldsymbol{s}$ to $\boldsymbol{s}_{\min}$ in $R$.
Then, path from $\min_{QH}(\boldsymbol{s},\boldsymbol{t})$ to $\min_{QH}(\boldsymbol{s}_{\min},\boldsymbol{t})$ is also in $R$.
This path is a path from $\boldsymbol{t}$ to $\boldsymbol{s}_{\min}$, since we have $\boldsymbol{s} \ge_{QH} \boldsymbol{t}$ and $\boldsymbol{s}_{\rm min} \le_{QH} \boldsymbol{t}_{\rm min}$ by assumption.
Then $\boldsymbol{t}$ and $\boldsymbol{s}_{\min}$ are in the same component, implying that $\boldsymbol{t}$ and $\boldsymbol{s}$ are in the same component.
%This contradicts that $\boldsymbol{t}$ to $\boldsymbol{s}$ are disconnected.
This contradicts disconnectivity of $\boldsymbol{t}$ and $\boldsymbol{s}$.
Therefore, $\boldsymbol{s}_{\rm min} \not\le_{QH} \boldsymbol{t}$, and Condition $C_3$ holds.
\end{proof}

From Claims~\ref{cl:certificate} and \ref{cl:exists-certificate},
we obtain a polynomial certificate for disconnectivity of $\boldsymbol{s}$ and $\boldsymbol{t}$.
Indeed, the size of the certificate is bounded in polynomial in the size of the input, and
we can check if conditions from C0 to C3 are satisfied by the certificate in polynomial time.
Therefore, $\aP(1)$ reconfiguration is in coNP.
\end{proof}

We then show our hardness result for $\aP(1)$ reconfiguration.

\begin{proposition}
The reconfiguration problem of $\aP(1)$ is weakly coNP-hard.
\end{proposition}

\begin{proof}
	We show this by a reduction from weak partition.
	Here, weak partition is, given distinct positive integers $a_1, \dots, a_\ell$, to determine if
	there exists a nonzero integer vector $\boldsymbol{x} \in \{-1,0,1\}^\ell \setminus \{\boldsymbol{0}\}$ such that
	$\sum_{i=1}^{\ell}a_ix_i = 0$.
	It is known that weak partition is weakly NP-hard~\cite{Sha79}.
	We first reduce weak partition to the following variant of partition problem, which we call \emph{bounded weak partition}.
	In bounded weak partition, we are given positive integers $b_1, \dots, b_n$ and
	asked to determine if there exists a nonzero integer vector $\boldsymbol{y} \in \{-1,0,1,2, \dots, 3n\}^n \setminus \{\boldsymbol{0}\}$ such that $\sum_{i=1}^{n}b_iy_i = 0$.
	We then reduce bounded weak partition to the 
	\emph{feasibility} problem of monotone quadratic ILS (i.e., ILS that is Horn and TVPI at the same time), using the idea of the proof of Theorem C in~\cite{Lag85}, 
	which shows the NP-hardness of monotone quadratic ILS\footnote{We note that \emph{good simultaneous approximation} is shown to be NP-hard in~\cite{Lag85}, 
		however, it is described how to formulate good simultaneous approximation as monotone quadratic ILS in the last part of~\cite{Lag85}.}.
	Finally, we reduce this feasibility problem to 
	the reconfiguration problem of $\aP(1)$.
	%Here, the positive value $3n$ that each $y_i$ can take in bounded weak partition is a key to the final reduction.
	We can then show that 
	an instance of weak partition is a yes instance 
	if and only if 
	the instance of the reconfiguration problem of $\aP(1)$ is a \emph{no} instance.
	This yields the weak coNP-hardness of the reconfiguration problem of $\aP(1)$.

We first reduce an instance of weak partition to an instance of bounded weak partition defined above.
Given an instance $I=(a_1,\dots, a_\ell)$ of weak partition with $\ell > 53$,
we can reduce it to bounded weak partition as follows.
Let $W = \sum_{i=1}^{\ell} a_i$.
Then we define an instance $I'=(b_1,\dots, b_n)$ of bounded weak partition as follows.
Set $n=2\ell$, $b_i = a_i + W^i$ for $i=1,\dots, \ell$, and $b_{i} = W^i$ for $i=\ell+1,\dots, 2\ell$.
We can then show that $I$ is a yes instance if and only if $I'$ is a yes instance.
Indeed, let $\boldsymbol{x}$ be a solution to $I$.
Define $\boldsymbol{y}$ as $y_i = x_i$ for $i=1,\dots, \ell$ and $y_i = -x_i$ for all $i=\ell+1,\dots, 2\ell(=n)$.
Then we have $\sum_{i=1}^{n}b_iy_i = \sum_{i=1}^{\ell}(a_i+W^i)x_i - \sum_{i=1}^{\ell}W^ix_i = \sum_{i=1}^{\ell}a_ix_i = 0$.
Therefore, $I'$ also has a solution.
For the other direction, let $\boldsymbol{y}$ be a solution to $I'$.
We first observe that $y_i = -y_{\ell+i}$ (or, equivalently, $y_i + y_{\ell+i} =0$) holds for all $i=1,\dots, \ell$.
Assume otherwise that $y_i \neq -y_{\ell+i}$ for some $i \in \{1,\dots, \ell \}$ and let $k$ be the maximum index among such indices.
Then $\sum_{i=1}^{n}b_iy_i = \sum_{i=1}^{k-1}(b_iy_i + b_{\ell+i}y_{\ell+i}) + (b_ky_k + b_{\ell+k}y_{\ell+k}) = \sum_{i=1}^{k-1}(b_iy_i + b_{\ell+i}y_{\ell+i}) + a_ky_k + (y_k + y_{\ell+k})W^k$.
We now show that $|(y_k + y_{\ell+k})W^k| > |\sum_{i=1}^{k-1}(b_iy_i + b_{\ell+i}y_{\ell+i}) + a_ky_k|$, which implies that $\sum_{i=1}^{n}b_iy_i \neq 0$, i.e., $\boldsymbol{y}$ is not a solution.
Indeed,
\begin{equation*}
\begin{array}{lll}
|\sum_{i=1}^{k-1}(b_iy_i + b_{\ell+i}y_{\ell+i}) + a_ky_k| &\le & \sum_{i=1}^{k-1}(|b_iy_i| + |b_{\ell+i}y_{\ell+i}|) + |a_ky_k|\\
&\le& 3n(\sum_{i=1}^{k-1}(|b_i| + |b_{\ell+i}|) + a_k)\\
&=& 3n(\sum_{i=1}^{k-1}(a_i+W^i + W^i) + a_k)\\
&=& 3n(\sum_{i=1}^{k-1}(2W^i) + \sum_{i=1}^{k}a_i)\\
&\le& 3n(2\frac{W^{k}-1}{W-1} + W)\\
&\le& 3n(2\frac{W^{k}}{W/2} + W)\\
&\le& 13nW^{k-1}.
\end{array}
\end{equation*}
Now, $13nW^{k-1} < W^k$ holds if $13n < W$, i.e., $26\ell < W$.
Note that $W = \sum_{i=1}^{\ell}a_{\ell} \ge \sum_{i=1}^{\ell}i = \frac{\ell(\ell-1)}{2}$ since $a_i$'s are distinct.
Hence, $26\ell < W$ holds if $26\ell < \frac{\ell(\ell-1)}{2}$, i.e., if $\ell >53$.
Therefore, we have $y_i = -y_{\ell+i}$ (or, equivalently, $y_i + y_{\ell+i} =0$) holds for all $i=1,\dots, \ell$.
Then we have $0 = \sum_{i=1}^{n}b_iy_i = \sum_{i=1}^{\ell}(b_iy_i - b_{\ell+i}y_i) = \sum_{i=1}^{\ell}((a_i+W^i)y_i - W^iy_i) = \sum_{i=1}^{\ell}a_iy_i$.
Moreover, since $\boldsymbol{y} \in \{-1,0,1,2, \dots, 3n\}^n \setminus \{\boldsymbol{0}\}$,
we obtain that $y \in \{-1,0,1\}^n \setminus \{\boldsymbol{0}\}$.
Therefore, if we define $\boldsymbol{x}$ as $x_i = y_i$ for $i=1,\dots, \ell$,
then it is a solution to $I$.
This complete the proof.
Note that %For the instance of bounded weak partition,
from the discussion above any solution $\boldsymbol{y}$ to $I'$ satisfies
(i) $\boldsymbol{y} \in \{-1,0,1\}^n \setminus \{\boldsymbol{0}\}$ and (ii) $\sum_{i=1}^{n}y_i = 0$.
These properties are crucial for our reduction to $\aP(1)$ reconfiguration.

We then reduce the instance $I'$ of bounded weak partition to $\aP(1)$ reconfiguration.
For this, we use the reduction by Lagarias~\cite{Lag85},
which shows the NP-hardness of monotone quadratic systems,
where an ILS is \emph{monotone quadratic} if each inequality has at most one positive coefficient and at most one negative coefficient
\footnote{We note that \emph{good simultaneous approximation} is shown to be NP-hard in~\cite{Lag85}, 
	however, it is described how to formulate good simultaneous approximation as monotone quadratic systems in the last part of~\cite{Lag85}.}.
For the sake of clarity, we use notation used in the proof of Theorem C in~\cite{Lag85} except that $Z$ in~\cite{Lag85} is replaced by $z$.
From the proof in~\cite{Lag85}, it can be seen that
%bounded weak partition
$I'$ has a solution if and only if the following ILS has a solution.
\if0
\begin{equation}\label{eq:Lagarias-original}
\left\{
\begin{array}{ll}
x_0-\frac{1}{p_0^R}z = 0 & \\
x_i -\frac{\theta^*_i}{Q_i^T}z \ge -\frac{1}{Q_i^T} & (i=1, \dots, n) \\
- x_i + \frac{\theta^*_i}{Q_i^T}z \ge - \frac{1}{Q_i^T} & (i=1, \dots, n) \\
1 \le z \le \sum_{i=1}^{n}\theta_i.
\end{array}
\right.
\end{equation}
\fi
\begin{equation}\label{eq:Lagarias}
\left\{
\begin{array}{ll}
x_0-\frac{1}{p_0^R}z = 0 & \\
x_i -\frac{\theta^*_i}{Q_i^T}z \ge -\frac{1}{Q_i^T} & (i=1, \dots, n) \\
- x_i + \frac{\theta^*_i}{Q_i^T}z \ge - \frac{3n}{Q_i^T} & (i=1, \dots, n) \\
1 \le z \le \sum_{i=1}^{n}\theta_i.
\end{array}
\right.
\end{equation}
Here, $p_0, Q_1,\dots, Q_n$ are primes and $R, T, \theta_1, \dots, \theta_n, \theta^*_{1}, \dots, \theta^*_n$ are integers 
determined by $b_i$'s in bounded weak partition.

Then our instance $I''$ of $\aP(1)$ reconfiguration is defined as follows.
Our ILS is given by

\begin{equation}\label{eq:Z=1}
\left\{
\begin{array}{ll}
x_0-\frac{1}{p_0^R}z \ge 0 & \\
 - x_0 + \frac{1}{p_0^R}z \ge -1 & \\
x_i -\frac{\theta^*_i}{Q_i^T}z \ge -\frac{1}{Q_i^T} & (i=1, \dots, n) \\
- x_i + \frac{\theta^*_i}{Q_i^T}z \ge -1 - \frac{\theta^*_i-2}{Q_i^T} & (i=1, \dots, n) \\
2n\frac{p_0^R}{Q_1^T}(x_0-\frac{1}{p_0^R}z) + \sum_{i=1}^{n}\left( x_i - \frac{\theta^*_i}{Q_i^T}z \right) \ge \frac{n+1}{2Q_1^T}, &
%2n\frac{p_0^R}{Q_1^T}x_0 + \sum_{i=1}^{n}x_i - (\frac{2n}{Q_1^T} + \sum_{i=1}^{n}\frac{\theta^*_i}{Q_i^T})z \ge \frac{n+1}{2Q_1^T}, &
\end{array}
\right.
\end{equation}
and we set
\begin{equation}
\begin{array}{ccl}
\boldsymbol{s}&=&(x_0,x_1,\dots, x_n,z) \\
&=& \left(\left\lceil \frac{\sum_{i=1}^{n}\theta_i}{p_0^R} \right\rceil, \left\lceil \frac{\theta_1^*\sum_{i=1}^{n}\theta_i}{Q_1^T} - \frac{1}{Q_1^T} \right\rceil, \dots, \left\lceil \frac{\theta_n^*\sum_{i=1}^{n}\theta_i}{Q_n^T} - \frac{1}{Q_n^T} \right\rceil, \sum_{i=1}^{n}\theta_i \right)
\end{array}
\end{equation}
and $\boldsymbol{t}=\boldsymbol{0}$.
Note that the last inequality of system~\eqref{eq:Z=1} is equivalent to
\begin{equation}
2n\frac{p_0^R}{Q_1^T}x_0 + \sum_{i=1}^{n}x_i - \left(\frac{2n}{Q_1^T} + \sum_{i=1}^{n}\frac{\theta^*_i}{Q_i^T}\right)z \ge \frac{n+1}{2Q_1^T}
\end{equation}
and thus system~\eqref{eq:Z=1} is dual Horn,
i.e., each inequality has at most one negative coefficient.
We actually have that the complexity index of system~\eqref{eq:Z=1} is exactly one.

We now show that $I'$ is a yes instance if and only if $I''$ is a no instance,
namely,
$I'$ has a solution
if and only if
$\boldsymbol{s}$ and $\boldsymbol{t}$ are \emph{disconnected} in the solution graph of system~\eqref{eq:Z=1}.

First, observe that if we do not have the last inequality in system~\eqref{eq:Z=1},
then we can reach from $\boldsymbol{s}$ to $\boldsymbol{t}$.
Indeed, we can reach from $(s_1,s_{n+2})$ to $(0,0)$ in the solution graph of
the first and second inequalities in system~\eqref{eq:Z=1}.
For each $i=1, \dots, n$,
we can also reach from $(s_i,s_{n+2})$ to $(0,0)$ in the solution graph of
the third and fourth inequalities in system~\eqref{eq:Z=1}.
Since $x_i's$ do not have inequality in common except the last inequality in system~\eqref{eq:Z=1},
we can reach from $\boldsymbol{s}$ to $\boldsymbol{t}$ if
we do not have the last inequality in system~\eqref{eq:Z=1}.
%We choose such a path arbitrarily and name it path $P$.

Now, suppose that $I'$ has a solution $\boldsymbol{y}$.
Then,
%it is shown in~\cite{Lag85} that
system~\eqref{eq:Lagarias} has
an integer solution $\boldsymbol{u} = (x_0, x_1, \dots, x_n,z)$.
We first observe that any path from $\boldsymbol{s}$ to $\boldsymbol{t}$ in $G(I'')$ 
without the last inequality in system~\eqref{eq:Z=1}
must path through $\boldsymbol{u}$.
This is because for a fixed value of $z$,
there exists exactly one value for $x_0$ (resp., for each $x_i$) that satisfies the first inequality (resp., both the second and third inequalities) in system~\eqref{eq:Lagarias}.
Indeed, this is clear for $x_0$.
For $x_i$ ($i=1, \dots, n$), from the second and third inequalities in system~\eqref{eq:Lagarias},
we have
$-\frac{1}{Q_i^T} \le x_i -\frac{\theta^*_i}{Q_i^T}z \le \frac{3n}{Q_i^T} < 1 - \frac{1}{Q_i^T}$,
where the last inequality follows from the size condition of $Q_i$ such that $Q_i^T > Q_1^T \ge 4(n+1)p_0^R$.
Therefore, the interval of $x_i$ is shorter than one for a fixed $z$, and
thus we have exactly one integer in the interval.
%the third and fourth inequalities in system~\eqref{eq:Z=1} forces to path through
%any point.
Second, from our reduction to bounded weak partition,
we have (i) $-\frac{1}{Q_i^T} \le x_i - \frac{\theta^*_i}{Q_i^T}z \le \frac{1}{Q_i^T}$ and
(ii) $|\{ i \mid x_i - \frac{\theta^*_i}{Q_i^T}z = -\frac{1}{Q_i^T} \}| = |\{ i \mid x_i - \frac{\theta^*_i}{Q_i^T}z = \frac{1}{Q_i^T} \}|$.
Therefore, it holds that
\begin{equation}
\begin{array}{lll}
\sum_{i=1}^{n}\left( x_i - \frac{\theta^*_i}{Q_i^T}z \right) &=&\sum_{i: y_i=1}\frac{1}{Q_i^T} - \sum_{i: y_i=-1}\frac{1}{Q_i^T}\\
&\le&\frac{n}{Q_1^T} - \frac{n}{Q_n^T}\\
&\le&\frac{n}{Q_1^T} - \frac{n}{2Q_1^T}\\
&=& \frac{n}{2Q_1^T},\\
&<& \frac{n+1}{2Q_1^T},
\end{array}
\end{equation}
where we use
$\frac{1}{Q_n^T} < \frac{1}{Q_{n-1}^T} < \dots < \frac{1}{Q_1^T} < \frac{2}{Q_n^T}$.
Thus, the last inequality in system~\eqref{eq:Z=1} is not satisfied.
Therefore, $\boldsymbol{s}$ cannot reach to $\boldsymbol{t}$ in the solution graph of system~\eqref{eq:Z=1}.

We then show the converse holds.
Suppose that $I'$ does not have a solution.
Let $P$ be an $\boldsymbol{s}$-$\boldsymbol{t}$ path in $G(I)$ without last inequality in system~\eqref{eq:Z=1}.
We show that all the vectors in $P$
satisfy the last inequality in system~\eqref{eq:Z=1},
implying that the $\boldsymbol{s}$-$\boldsymbol{t}$ path is also a path in system~\eqref{eq:Z=1}.
Let $\boldsymbol{u} = (x_0, x_1, \dots, x_n,z)$ be an arbitrary vector in $P$.
If $x_0-\frac{1}{p_0^R}z > 0$ holds, then the last inequality in system~\eqref{eq:Z=1} is satisfied since we have
\begin{equation*}
\begin{array}{lll}
2n\frac{p_0^R}{Q_1^T}(x_0-\frac{1}{p_0^R}z) + \sum_{i=1}^{n}\left( x_i - \frac{\theta^*_i}{Q_i^T}z \right) &\ge& 2n\frac{p_0^R}{Q_1^T}\frac{1}{p_0^R} - \sum_{i=1}^{n}\frac{1}{Q_i^T}\\
&\ge& \frac{2n}{Q_1^T} - \sum_{i=1}^{n}\frac{1}{Q_1^T}\\
&=& \frac{n}{Q_1^T}\\
&\ge& \frac{n+1}{2Q_1^T}.
\end{array}
\end{equation*}
Hence, assume that $x_0-\frac{1}{p_0^R}z = 0$ holds.
Recall that if $I'$ has no solution, then system~\eqref{eq:Lagarias} has no integer solution.
Therefore,
%$y_i \ge 3n+1$ for at least one $j$ in $\{ 1, \dots, n \}$ in bounded weak partition,
%since otherwise it has a solution.
%Namely, for any value of $z$,
at least one $j$ satisfies that
$x_j - \frac{\theta_j^*}{Q_j^T}z \ge \frac{3n+1}{Q_j^T}$.
Thus, it follows that
\begin{equation}
\begin{array}{lll}
2n\frac{p_0^R}{Q_1^T}(x_0-\frac{1}{p_0^R}z) + \sum_{i=1}^{n}\left( x_i - \frac{\theta^*_i}{Q_i^T}z \right) &=&\frac{3n+1}{Q_j^T} - \sum_{i\neq j, 1 \le i \le n}\frac{1}{Q_i^T} \\
&\ge& \frac{3n+1}{2Q_1^T} - \frac{n-1}{Q_1^T}\\
&=& \frac{n+3}{2Q_1^T},\\
&\ge& \frac{n+1}{2Q_1^T},
\end{array}
\end{equation}
where the second inequality holds since
$\frac{1}{Q_n^T} < \frac{1}{Q_{n-1}^T} < \dots < \frac{1}{Q_1^T} < \frac{2}{Q_n^T}$.
Therefore, the last inequality in system~\eqref{eq:Z=1} is satisfied by $\boldsymbol{u}$.
%any $\boldsymbol{s}$-$\boldsymbol{t}$ path.
This completes the proof.
\end{proof}

\section{Tractable subclasses of $Z(I) = 1$}
\label{sec:Z=1-tractable}
In this section, we show that
certain subclasses of $\aP(1)$ reconfiguration are solvable in polynomial time.
%UTVPI reconfiguration and $\aP$ reconfiguration with fixed number of variables are both  solvable in polynomial time

\subsection{Unit systems}
An integer linear system $I = (A,\boldsymbol{b},d)$ is called \emph{unit} if $A \in \{ 0, \pm 1 \}^{m \times n}$ holds for positive integers $m$ and $n$.
For the feasibility problem, it is known that ILS$(1)$ restricted to unit integer linear systems, denoted by unit ILS(1), is polynomially solvable~\cite{KiM16}.
In this subsection, we consider the reconfiguration problem of
%this class of integer linear systems, i.e., unit integer linear systems with $Z(I) \le 1$.
unit ILS(1).
We note that unit ILS(1) includes a well-studied subclass of ILS such as
unit Horn ILS (e.g.,~\cite{Cha84,ChS13,SuW15}) and unit TVPI (UTVPI) ILS (e.g., \cite{JMS94,Sub04,LaM05}).
In this subsection, we will show the following theorem.

\begin{theorem}\label{thm:UILS-Z=1-P}
The reconfiguration problem of unit ILS(1) is polynomially solvable.
\end{theorem}

Any ILS in unit ILS(1) can be decomposed to a unit Horn ILS and a UTVPI ILS, 
using the decomposition of general ILS(1).
We thus first provide polynomial time algorithms to solve the reconfiguration problems of these two ILSes and
then combine them to solve that of unit ILS(1).

We first show the following result.

\begin{proposition}\label{thm:UHorn-P}
The reconfiguration problem of unit Horn ILS is polynomially solvable.
\end{proposition}

\begin{proof}
	Given a Horn ILS $I=(A,\boldsymbol{b},d)$ and a solution $\boldsymbol{s}$ of $I$,
	we provide an algorithm to obtain a unique minimal solution $\boldsymbol{s}_{\min}$ in the same component as $\boldsymbol{s}$ in polynomial time.
	Once this is done, the reconfiguration problem can be solved in polynomial time
	by checking whether unique minimal solutions in the same components as given two solutions are the same or not.

\if0 
For two vectors $\boldsymbol{u}$ and $\boldsymbol{v}$,
define $[\boldsymbol{u},\boldsymbol{v}]:=\{ \boldsymbol{w} \mid \boldsymbol{u} \le \boldsymbol{w} \le \boldsymbol{v} \}$.
Our algorithm first finds a unique minimal solution $\boldsymbol{s}_{\min}^1$ in $[\boldsymbol{s} - \boldsymbol{1}, \boldsymbol{s}] \cap R(I)$ that is in the same component as $\boldsymbol{s}$, 
where $R(I) = \{ \boldsymbol{x} \mid A\boldsymbol{x} \ge \boldsymbol{b} \}$.
We note that since $s_j-1 \le x_j \le s_j$ for each $j \in \{1, \dots, n\}$ is Horn,
we also have a unique minimal solution in $[\boldsymbol{s} - \boldsymbol{1}, \boldsymbol{s}] \cap R(I)$.
Then, we (possibly exponentially) decrease simultaneously the values of those variables $j$ for which $s_j$ and $s_{\min}^1(j)$ differs 
	until some variables coincides with their lower bounds.
	Then, we fix the values of these variables and iterate the above procedure for the ILS with less variables,
	eventually get to the unique minimal solution $\boldsymbol{s}_{\min}$ in at most $n$ iterations.
We formally describe the algorithm in Algorithm~\ref{alg:unit-Horn}.
\fi 

Consider the following Horn ILS $I'$ with matrix representation $\{ A\boldsymbol{x} \ge \boldsymbol{b}, \boldsymbol{x} \ge \boldsymbol{s}-\boldsymbol{1}, -\boldsymbol{x} \ge -\boldsymbol{s} \}$, where the second and last inequalities are equivalent to $s_j-1 \le x_j \le s_j$ for each $j \in \{1, \dots, n\}$.
%Let $I_1$ be a Horn ILS	
Our algorithm finds in $I'$ a unique minimal solution $\boldsymbol{s}_{\min}^1$ in the same component as $\boldsymbol{s}$.
	Then, we (possibly exponentially) decrease simultaneously the values of those variables $j$ for which $s_j$ and $s_{\min}^1(j)$ differs
	until some variables coincides with their lower bounds.
	Then, we fix the values of these variables and iterate the above procedure for the ILS with less variables,
	eventually get to the unique minimal solution $\boldsymbol{s}_{\min}$ in at most $n$ iterations.

We formally describe our algorithm in Algorithm~\ref{alg:unit-Horn}.
For a vector $\boldsymbol{u}$ we also denote by $u(j)$ the $j$-th component of $\boldsymbol{u}$.
For index set $V$ of variables and $U \subseteq V$,
let $\boldsymbol{e}_{U}$ be a vector such that
$\boldsymbol{e}_{U}(j) = 1$ if $j \in U$ and $\boldsymbol{e}_{U}(j) = 0$ otherwise (i.e., $j \in V \setminus U$).
For simplicity, we assume that inequalities $x_j \ge 0$ for $j \in \{ 1, \dots, n \}$ are included in $A\boldsymbol{x}\ge \boldsymbol{b}$.
We also assume that $\boldsymbol{b} \in \mathbb{Z}^m$ for simplicity,
however the result also holds for an arbitrary $\boldsymbol{b}$.
%\LinesNotNumbered
\LinesNumberedHidden
\begin{algorithm}

	\caption{Computing a unique minimal solution $\boldsymbol{s}_{\min}$ in the same component as a solution $\boldsymbol{s}$ in a unit Horn system $I=(A,\boldsymbol{b},d)$}
	\label{alg:unit-Horn}
	\ShowLn $V:=\{1, \dots, n\}$\\
	\ShowLn $T:=\emptyset,U:=\emptyset,W:=\emptyset$\\
	\ShowLn $\boldsymbol{u}:=\boldsymbol{s}$\\
	\ShowLn $R:=\{\boldsymbol{x} \in D^V \mid A\boldsymbol{x}\ge \boldsymbol{b} \}$\\
	\ShowLn $p:=0$\\
	\ShowLn \While{$T \neq V$}{
		\ShowLn Find a unique minimal solution $\boldsymbol{u}_{\min}$ in the same component as $\boldsymbol{u}$ in $G([\boldsymbol{u} - \boldsymbol{e}_{V\setminus T}, \boldsymbol{u}] \cap R)$\\
		%Find a unique minimal solution $\boldsymbol{u}_{\min}$ in the same component as $\boldsymbol{u}$ in $G([\boldsymbol{u} - \boldsymbol{e}_{V}, \boldsymbol{u}] \cap R)$\\
		\ShowLn $U \leftarrow \{ j \in V \setminus T \mid u_j \neq u_{\min}(j) \}$\\
		%$U \leftarrow \{ j \in V \mid u_j \neq u_{\min}(j) \}$\\
		%$p \leftarrow \max\{k \mid \boldsymbol{u} - k(\boldsymbol{u} - \boldsymbol{u}_{\min}) \in R \}$\\
		\ShowLn $p \leftarrow \max\{k \mid \boldsymbol{u} - k\boldsymbol{e}_U \in R \}$\\
		%$\boldsymbol{u} \leftarrow \boldsymbol{u} - p(\boldsymbol{u} - \boldsymbol{u}_{\min})$\\
		\ShowLn $\boldsymbol{u} \leftarrow \boldsymbol{u} - p\boldsymbol{e}_U$\\
		%$W \leftarrow \{ j \in V \mid u_j = 0\ \text{or}\ \exists i(a_{ij} = 1 \wedge \{ j' \in V \mid a_{ij'} = -1  \} \subseteq T \wedge \sum_{j \in V}a_{ij}u_j = b_i) \}$\\
		\ShowLn $W \leftarrow \{ j \in U \mid u_j = 0\ \text{or}$\\
		\ \ \ \ \ \ \ \ \ \  $\exists i(a_{ij} = 1 \wedge \{ j' \in V \mid a_{ij'} = -1  \} \subseteq V \setminus U \wedge \sum_{j \in V}a_{ij}u_j = b_i) \}$\\
		\setcounter{AlgoLine}{11}
		\ShowLn $T \leftarrow (V \setminus U) \cup W$
		%$T \leftarrow \{ j \in V \mid u_j = 0\ \text{or}\ \exists i(a_{ij} = 1 \wedge \{ j' \in V \mid a_{ij'} = -1  \} \subseteq T \wedge \sum_{j \in V}a_{ij}u_j = b_i) \}$\\
		%$\boldsymbol{b} \leftarrow \boldsymbol{b} - A[T] \boldsymbol{u}[T]$\\
		%$R \leftarrow \{\boldsymbol{x} \in D^V \mid A[V] \boldsymbol{x}_{V} \ge \boldsymbol{b}_{V} \}$
	}
	\setcounter{AlgoLine}{12}
	\ShowLn output $\boldsymbol{u}$ and halt
	%	\end{algorithmic}
\end{algorithm}
\LinesNumbered
Below, we show the correctness of Algorithm~\ref{alg:unit-Horn} and analyze its running time.

For the correctness of Algorithm~\ref{alg:unit-Horn},
we claim that
%in each iteration in the while loop,
at the end of each iteration in the while loop,
%(i) we have $u_j = s_{\min}(j)$ for all $j \in T$ at the end of the iteration and
(i) we have $u_j = s_{\min}(j)$ for all $j \in T$ and
%(ii) we can reach to $\boldsymbol{u}$ at the end of the iteration from $\boldsymbol{u}$ at the beginning of the iteration.
%(ii) we can reach from $\boldsymbol{s}$ to $\boldsymbol{u}$ at the end of the iteration in $G(R)$.
(ii) we can reach from $\boldsymbol{s}$ to $\boldsymbol{u}$ in $G(R)$.
Then, since $T=V$ holds at the end of the algorithm, we know that the output $\boldsymbol{u}$ of the algorithm equals to $\boldsymbol{s}_{\min}$.
We show the claim by induction on the number of iterations in the while loop.

Consider the first iteration.
Let $\boldsymbol{u}_{\rm begin}$ (resp., $\boldsymbol{u}_{\rm end}$) be the vector $\boldsymbol{u}$ at the beginning (resp., end) of the iteration.
For (i),
%$\boldsymbol{u}_{\min}$ is a unique minimal solution in the same component as $\boldsymbol{u}$ in $[\boldsymbol{u} - \boldsymbol{e}_{V}, \boldsymbol{u}] \cap R$.
%We want to show that $u_{\min}(j) = s_{\min}(j)$ for all $j \in T$.
assume otherwise that $u_{\rm end}(j) > s_{\min}(j)$ holds for some $j \in T$.
Consider a monotone $\boldsymbol{u}_{\rm begin}-\boldsymbol{s}_{\min}$ path $P$ in $G(R)$, which exists from Lemma~\ref{lem:Horn-monotone-path}.
Let $j \in T$ be the index such that the value of $x_j$ first decreases in path $P$ among the indices in $T$,
and let $\boldsymbol{v}$ in $P$ be the vector exactly before the value of $x_j$ decreases.
We show that such $j$ cannot exist.
Indeed, assume first that $j \in V \setminus U$.
Then there exists $i \in \{ 1, \dots, m \}$
such that $a_{ij} = 1$ and $\sum_{j'=1}^{n}a_{ij'}u_{\min}(j') = b_i$, since otherwise we can decrease the value of $x_j$ in $\boldsymbol{u}_{\min}$ without violating any inequality,
obtaining a contradiction that $\boldsymbol{u}_{\min}$ is a unique minimal solution.
Note that $a_{ij'} \le 0$ for $j' \neq j$ since $I$ is unit Horn.
Moreover, if $a_{ij'} = -1$ then $j' \in V \setminus U$ holds,
since otherwise we have $\sum_{j'=1}^{n}a_{ij'}u(j') < \sum_{j'=1}^{n}a_{ij'}u_{\min}(j') = b_i$, which contradicts that  $\boldsymbol{u}$ is a solution to $I$.
Therefore, we have $\sum_{j'=1}^{n}a_{ij'}u_{\min}(j') = \sum_{j' \in V \setminus U}a_{ij'}u_{\min}(j') + u_{\min}(j) = b_i$.
However, since $j \in T$ be the index such that the value of $x_j$ first decreases in path $P$ among the indices in $T$ and $U \setminus V \subseteq T$,
we have $\sum_{j' \in V \setminus U}a_{ij'}u_{\min}(j') + u_{\min}(j) = \sum_{j' \in V \setminus U}a_{ij'}v_{j'} + v_j = b_i$.
Therefore, we cannot decrease the value of $x_j$ in $\boldsymbol{v}$, obtaining a contradiction.
%Hence, $j \not\in U \setminus V$.
Assume next that $j \in W$.
Then by definition it follows that $u_{\rm end}(j) = 0$; or
there exists $i$ such that $a_{ij} = 1$, $\{ j' \in V \mid a_{ij'} = -1 \} \subseteq V \setminus U$ and $\sum_{j \in V}a_{ij}u_{\rm end}(j) = b_i$.
The former does not occur, since we cannot have $s_{\min}(j) < u_{\rm end}(j) = 0$.
The latter does not occur either,
since the latter implies that we have an inequality $x_j + \sum_{j' \in V \setminus U}a_{ij'}x_{j'} \ge b_i$ in the system
and that $u_{\rm end}(j) + \sum_{j' \in V \setminus U}a_{ij'}u_{\rm end}(j') = u_{\rm end}(j) + \sum_{j' \in V \setminus U}a_{ij'}s_{\min}(j') = b_i$,
since $u_{\rm end}(j') = s_{\min}(j')$ for $j' \in V \setminus U$.
Thus, we cannot have $s_{\min}(j) < u_{\rm end}(j)$.
Therefore, we have $u_{\rm end}(j) = s_{\min}(j)$ for all $j \in T$ and (i) is proven.

For (ii), let $P': \boldsymbol{u}_{\rm begin} = \boldsymbol{u}^0 \rightarrow \boldsymbol{u}^1
\rightarrow \dots \rightarrow \boldsymbol{u}^\ell = \boldsymbol{u}_{\min}$
be a monotone $\boldsymbol{u}_{\rm begin}$-$\boldsymbol{u}_{\min}$ path
in $G([\boldsymbol{u}_{\rm begin} - \boldsymbol{e}_{V\setminus T}, \boldsymbol{u}_{\rm begin}] \cap R)$,
which exists by Lemma~\ref{lem:Horn-monotone-path}.
We show that for each $1 \le q \le p-1$,
$P'_q:(\boldsymbol{u}^0 - q\boldsymbol{e}_U) \rightarrow (\boldsymbol{u}^1 - q\boldsymbol{e}_U)
\rightarrow \dots \rightarrow (\boldsymbol{u}^\ell- q\boldsymbol{e}_U)$ is
a $(\boldsymbol{u}_{\rm begin} - q\boldsymbol{e}_U)$-$(\boldsymbol{u}_{\rm begin} - (q+1)\boldsymbol{e}_U)$ path in $G(R)$.
Then we can reach from $\boldsymbol{u}_{\rm begin} (=\boldsymbol{s})$ to $(\boldsymbol{u}_{\rm begin} - \boldsymbol{e}_U)$,
from $(\boldsymbol{u}_{\rm begin} - \boldsymbol{e}_U)$ to $(\boldsymbol{u}_{\rm begin} - 2\boldsymbol{e}_U)$, $\cdots$,
from $(\boldsymbol{u}_{\rm begin} - (p-1)\boldsymbol{e}_U)$ to $(\boldsymbol{u}_{\rm begin} - p\boldsymbol{e}_U) = \boldsymbol{u}_{\rm end}$,
thus showing (ii).
To see that $P'_q$ is a path in $G(R)$ for each $q$,
we show that each vector in $P'_q$ is contained in $R$, that is, each vector in $P'_q$ is a solution to $I$.
Fix $q \in \{ 1,\dots, p-1 \}$ and $k \in \{1, \dots, \ell \}$, and consider vector $\boldsymbol{u}^k- q\boldsymbol{e}_U$.
%Arbitrarily choose the $i$-th inequality $\sum_{j=1}^{n}a_{ij}x_j\ge b_i$ of $I$.
Choose $i \in  \{ 1, \dots, m\}$ arbitrarily and consider the $i$-th inequality $\sum_{j=1}^{n}a_{ij}x_j\ge b_i$ of $I$.
%Since $\boldsymbol{u}^k$ is a solution, we have
%$\sum_{j=1}^{n}a_{ij}u^k_j\ge b_i$.
We examine two cases: (a) $a_{ij} \ge 0$ for all $j \in U$, (b) $a_{ij} = -1$ for some $j \in U$.
For (a), since path $P'_q$ is monotone and $q \le  p-1$,
we have $\boldsymbol{u}^k- q\boldsymbol{e}_U \ge \boldsymbol{u}^\ell- q\boldsymbol{e}_U \ge \boldsymbol{u}^\ell- (p-1)\boldsymbol{e}_U = \boldsymbol{u}_{\rm end}$.
Since $\boldsymbol{u}_{\rm end}$ is a solution by definition, it satisfies the inequality.
Therefore, $\sum_{j=1}^{n}a_{ij}(u^k_j - qe_U(j)) \ge \sum_{j=1}^{n}a_{ij}u_{\rm end}(j)\ge b_i$ holds,
and thus $\boldsymbol{u}^k - q\boldsymbol{e}_U$ satisfies the inequality.
For (b), note that there exists at most one $j'$ with $a_{ij'} = 1$ since $I$ is unit Horn.
Therefore, we have $\sum_{j=1}^{n}a_{ij}(u^k_j - qe_U(j)) = \sum_{j=1}^{n}a_{ij}u^k_j - \sum_{j \in U:a_{ij} = 1} qe_U(j) + \sum_{j \in U:a_{ij} = -1} qe_U(j) \ge \sum_{j=1}^{n}a_{ij}u^k_j \ge b_i$,
where the last inequality holds since $\boldsymbol{u}^k$ is a solution to $I$.
Thus, $\boldsymbol{u}^k- q\boldsymbol{e}_U$ satisfies the inequality.
Since $i$ is arbitrary, $\boldsymbol{u}^k- q\boldsymbol{e}_U$ satisfies all the inequalities of $I$.
Thus, $\boldsymbol{u}^k- q\boldsymbol{e}_U$ is a solution to $I$.
Hence, (ii) is proven.
%Fix $k$ and let $\boldsymbol{w}$ be a vector in $P'_k$.
%By definition, we have $\boldsymbol{u}_{\rm end} \le \boldsymbol{w} \le \boldsymbol{u}_{\rm begin}$.
%Consider an inequality $\sum_{j=1}^{n}a_{ij}x_j\ge b_i$.
%If $a_{ij} \ge 0$ for all $j$, then the inequality is satisfied by $\boldsymbol{w}$ since $\boldsymbol{w} \ge \boldsymbol{u}_{\rm end}$ and $\boldsymbol{u}_{\rm end} \in R$.
%If $a_{ij} \le 0$ for all $j$, then the inequality is satisfied by $\boldsymbol{w}$ since $\boldsymbol{w} \le \boldsymbol{u}_{\rm begin}$ and $\boldsymbol{u}_{\rm begin} \in R$.
%Assume otherwise that there exist $j$ with $a_{ij} > 0$ and $j'$ with $a_{ij'} < 0$.

We then consider the $r$-th iteration of the algorithm, where $r \ge 2$.
Let $\boldsymbol{u}_{\rm begin}$ (resp., $\boldsymbol{u}_{\rm end}$) be the vector $\boldsymbol{u}$ at the beginning (resp., end) of the iteration.
For (i), assume otherwise that $u_j > s_{\min}(j)$ holds for some $j \in T$.
Consider a monotone $\boldsymbol{u}_{\rm begin}-\boldsymbol{s}_{\min}$ path $P$ in $R$, which exists since $\boldsymbol{u}_{\rm begin}$ and $\boldsymbol{s}$ are
in the same component from the inductive hypothesis for (ii).
Then since for the indices in $T$ at the beginning of the $r$-th iteration the values of
$\boldsymbol{u}_{\rm begin}$ and $\boldsymbol{s}_{\min}$ are the same by inductive hypothesis,
we can use the same argument as that to show (i) for the case of $r=1$.
We can also prove (ii) in a similar way as the proof for $r=1$.
Therefore, (i) and (ii) hold at the end of every iteration by induction.

Finally, we analyze the running time of Algorithm~\ref{alg:unit-Horn}.
Since $T$ increases at least one in each iteration of the while loop,
the number of iterations is at most $n$.
In each iteration of the while loop,
we can find a unique minimal solution $\boldsymbol{u}_{\min}$ in polynomial time
by using the pseudo-polynomial time algorithm for Horn systems in Subsection~\ref{subsec:Horn}.
Moreover, $p$ can be computed in polynomial time
since for each inequality $\sum_{j=1}^{n}a_{ij}x_j \ge b_i$
we can compute how much the left hand side changes if we change $\boldsymbol{x}$ to $\boldsymbol{x} -\boldsymbol{e}_U$.
Therefore, each iteration of the while loop can be done in polynomial time.
Hence, Algorithm~\ref{alg:unit-Horn} runs in polynomial time.
This completes the proof.
\end{proof}

We next show the following result.

\begin{proposition}\label{thm:UTVPI-P}
The reconfiguration problem of UTVPI ILS is polynomially solvable.
\end{proposition}

We first note that if $\boldsymbol{t}$ is reachable from $\boldsymbol{s}$,
then there exists a $\boldsymbol{t}$-monotone path by Lemma~\ref{lem:TVPI-monotone-path}.
Then, by replacing $x_j$ with $x_j'=d-x_j$ for each $j$ with $s_j < t_j$,
we obtain a monotone path from $\boldsymbol{s}'$ to $\boldsymbol{t}'$, where we note that $\boldsymbol{s}' \ge \boldsymbol{t}'$ holds.
Therefore, we assume that $\boldsymbol{s} \ge \boldsymbol{t}$ holds without loss of generality.
%Without loss of generality, we assume that $\boldsymbol{s} \le \boldsymbol{t}$.

The following lemma is a key lemma to prove Proposition~\ref{thm:UTVPI-P}.
Let $p = |\boldsymbol{s} - \boldsymbol{t}|_{\infty}$, i.e., $p = \max_{j} |s_j - t_j|$.
Define $\boldsymbol{u}^i:= \max(\boldsymbol{s} - i\mathbf{1},\boldsymbol{t})$ for $0 \le i \le p$.
We note that $\boldsymbol{u}^0 = \boldsymbol{s}$ and $\boldsymbol{u}^p = \boldsymbol{t}$ hold.

\begin{lemma}\label{lem:UTVPI-connectivity-equivalece}
Let $I$ be an instance UTVPI system and $\boldsymbol{s},\boldsymbol{t}$ be solutions of $I$.
Then the following are equivalent:
\begin{description}
\item[(1)] $\boldsymbol{t}$ is reachable from $\boldsymbol{s}$ in $G(I)$.
\item[(2)] $\boldsymbol{u}^{i+1}$ is reachable from $\boldsymbol{u}^i$ in $G(I)$ for all $0 \le i \le p-1$.
\item[(3)] $\boldsymbol{u}^1$ is reachable from $\boldsymbol{u}^0$ in $G(I)$.
\end{description}
\end{lemma}

\begin{proof}
We show the following chain of implications: (1) $\Rightarrow$ (3) $\Rightarrow$ (2) $\Rightarrow$ (1).
Note that (2) $\Rightarrow$ (1) is straightforward.
Hence, we show (1) $\Rightarrow$ (3) and (3) $\Rightarrow$ (2) in the following.

We first show (1) $\Rightarrow$ (3).
Since $\boldsymbol{t}$ is reachable from $\boldsymbol{s}$,
there exists an $\boldsymbol{s}$-$\boldsymbol{t}$ path in $G(I)$.
Let such a path be $\boldsymbol{s}=\boldsymbol{s}^0 \rightarrow \boldsymbol{s}^1 \rightarrow \dots \rightarrow \boldsymbol{s}^\ell = \boldsymbol{t}$, denoted by $P$.
%From lemma~\ref{lem:TVPI-monotone-path}, we can assume that $P$ is $\boldsymbol{t}$-monotone.
%Since we are assuming that $\boldsymbol{s} \le \boldsymbol{t}$ holds,
%$\boldsymbol{t}$-monotonicity of $P$ especially implies that $\boldsymbol{s} \le \boldsymbol{s}^i \le \boldsymbol{t}$ for all $0 \le i \le \ell$.
From the discussion before the lemma,
we assume that $P$ is monotone.
Thus, $\boldsymbol{s} \ge \boldsymbol{s}^i \ge \boldsymbol{t}$ for all $0 \le i \le \ell$.
Now, consider a sequence $\max(\boldsymbol{s}^0,\boldsymbol{u}^1) \rightarrow \max(\boldsymbol{s}^1,\boldsymbol{u}^1) \rightarrow \dots \rightarrow \max(\boldsymbol{s}^\ell,\boldsymbol{u}^1)$.
We show that this sequence is indeed a path from $\boldsymbol{u}^0$ to $\boldsymbol{u}^1$ in $G(I)$, which implies that (3) holds.
First, we have $\max(\boldsymbol{s}^0,\boldsymbol{u}^1) =\max(\boldsymbol{s},\max(\boldsymbol{s}-\mathbf{1},\boldsymbol{t})) = \boldsymbol{s} =\boldsymbol{u}^0$ and
$\max(\boldsymbol{s}^\ell,\boldsymbol{u}^1) = \max(\boldsymbol{t},\max(\boldsymbol{s}-\mathbf{1},\boldsymbol{t})) = \max(\boldsymbol{s}-\mathbf{1},\boldsymbol{t}) = \boldsymbol{u}^1$.
Next, we have $\dist(\max(\boldsymbol{s}^i,\boldsymbol{u}^1), \max(\boldsymbol{s}^{i+1},\boldsymbol{u}^1)) \le 1$ for all $0 \le i \le \ell-1$,
since $\boldsymbol{s}^i$ and $\boldsymbol{s}^{i+1}$ differs at one position.
Finally, we show that $\max(\boldsymbol{s}^i,\boldsymbol{u}^1)$ is a solution to $I$ for all $0 \le i \le \ell$.
Fix $i \in \{ 0, \dots, \ell\}$.
We show that each inequality of $I$ is satisfied by $\max(\boldsymbol{s}^i,\boldsymbol{u}^1)$.
We have three types of inequalities to consider, namely, $x_q + x_r \ge c$, $x_q - x_r \ge c$, and $-x_q - x_r \ge c$.
Observe first that $x_q + x_r \ge c$ is satisfied by $\max(\boldsymbol{s}^i,\boldsymbol{u}^1)$,
since $\max(\boldsymbol{s}^i,\boldsymbol{u}^1) \ge \boldsymbol{t}$ and $\boldsymbol{t}$ satisfies the inequality.
Likewise, $-x_q - x_r \ge c$ is satisfied by $\max(\boldsymbol{s}^i,\boldsymbol{u}^1)$,
since $\max(\boldsymbol{s}^i,\boldsymbol{u}^1) \le \boldsymbol{s}$ and $\boldsymbol{s}$ satisfies the inequality.
For $x_q - x_r \ge c$,
we show that $\boldsymbol{u}^1$ satisfies this inequality.
Then $\max(\boldsymbol{s}^i,\boldsymbol{u}^1)_q +\max(\boldsymbol{s}^i,\boldsymbol{u}^1)_r \ge c$ follows since $x_q - x_r \ge c$ is closed under $\max$ operation.
To show that $\boldsymbol{u}^1 (= \max(\boldsymbol{s} - \mathbf{1},\boldsymbol{t}))$ satisfies $x_q - x_r \ge c$,
assume first that $s_r-1 \le t_r$ holds.
Then we have $u_q^1 -u_r^1 = \max(s_q-1,t_q) - \max(s_r-1,t_r) = \max(s_q-1,t_q) - t_r \ge t_q -t_r \ge c$.
Next, assume otherwise that $s_r-1 > t_r$ holds.
Then if $s_q-1 \ge t_q$, then we have $u_q^1 -u_r^1 = s_q-1 - (s_r-1) = s_q -s_r \ge c$,
and if $s_q-1 < t_q$, then we have $s_q = t_q$ since $s_q \ge t_q$ and thus $u_q^1 -u_r^1 = t_q - (s_r-1) = s_q - (s_r-1) \ge s_q - s_r \ge c$.
Therefore, $\boldsymbol{u}^1$ satisfies $x_q - x_r \ge c$.
This complete the proof for (1) $\Rightarrow$ (3).

We next show (3) $\Rightarrow$ (2).
We show that $\boldsymbol{u}^{i+1}$ is reachable from $\boldsymbol{u}^i$ in $G(I)$ for all $0 \le i \le p-1$ by induction on $i$.
For $i=0$, $\boldsymbol{u}^{1}$ is reachable from $\boldsymbol{u}^0$ by (3).
For $i>0$, assume that there exists a monotone $\boldsymbol{u}^{i-1}$-$\boldsymbol{u}^{i}$ path in $G(I)$.
Let such a path be $\boldsymbol{u}^{i-1}=\boldsymbol{u}^{i-1,0} \rightarrow \boldsymbol{u}^{i-1,1} \rightarrow \dots \rightarrow \boldsymbol{u}^{i-1,\ell} = \boldsymbol{u}^i$.
Define $\boldsymbol{u}^{i,k} := \max((\boldsymbol{u}^{i-1,k} - \mathbf{1}), \boldsymbol{t})$ for $0 \le k \le  \ell$.
We show that $\boldsymbol{u}^{i,0} \rightarrow \boldsymbol{u}^{i,1} \rightarrow \dots \rightarrow \boldsymbol{u}^{i,\ell}$ is a $\boldsymbol{u}^i$-$\boldsymbol{u}^{i+1}$path in $G(I)$.

First, we have
\begin{align*}
\boldsymbol{u}^{i,0} &= \max((\boldsymbol{u}^{i-1,0} - \mathbf{1}), \boldsymbol{t}) \\
&= \max((\boldsymbol{u}^{i-1} - \mathbf{1}), \boldsymbol{t}) \\
&= \max(\max(\boldsymbol{s}-(i-1)\mathbf{1},\boldsymbol{t}) - \mathbf{1}, \boldsymbol{t}) \\
&= \max(\max(\boldsymbol{s}-i\mathbf{1},\boldsymbol{t} - \mathbf{1}), \boldsymbol{t}) \\
&= \max((\boldsymbol{s}-i\mathbf{1}), \boldsymbol{t}) = \boldsymbol{u}^i
\end{align*}
 and, similarly,
$\boldsymbol{u}^{i,\ell} = \max((\boldsymbol{u}^{i-1,\ell} - \mathbf{1}), \boldsymbol{t}) = \max((\boldsymbol{u}^i - \mathbf{1}), \boldsymbol{t}) = \boldsymbol{u}^{i+1}$.

Next, it follows that $\dist(\boldsymbol{u}^{i,k}, \boldsymbol{u}^{i,k+1}) \le 1$ for all $0 \le k \le \ell-1$,
since $\boldsymbol{u}^{i-1,k}$ and $\boldsymbol{u}^{i-1,k+1}$ differs at one position.

Finally, we show that $\boldsymbol{u}^{i,k}$ is a solution to $I$ for all $0 \le k \le \ell$.
Fix $k \in \{ 0, \dots, \ell\}$.
We show that each inequality of $I$ is satisfied by $\boldsymbol{u}^{i,k}$.
We have three types of inequalities to consider, namely, $x_q + x_r \ge c$, $x_q - x_r \ge c$, and $-x_q - x_r \ge c$.
Observe first that $x_q + x_r \ge c$ is satisfied by $\boldsymbol{u}^{i,k}$, since $\boldsymbol{u}^{i,k} \ge \boldsymbol{t}$ and $\boldsymbol{t}$ satisfies the inequality.
Likewise, $-x_q - x_r \ge c$ is satisfied by $\boldsymbol{u}^{i,k}$, since $\boldsymbol{u}^{i,k} \le \boldsymbol{s}$ and $\boldsymbol{s}$ satisfies the inequality.
Here, $\boldsymbol{u}^{i,k} (= \max((\boldsymbol{u}^{i-1,k} - \mathbf{1}),\boldsymbol{t}) \le \boldsymbol{s}$ holds
since we have $\boldsymbol{t} \le \boldsymbol{s}$ and $\boldsymbol{u}^{i-1,k} \le \boldsymbol{u}^{i-1} = \max(\boldsymbol{s} - (i-1)\mathbf{1},\boldsymbol{t}) \le \boldsymbol{s}$ from monotonicity of the $\boldsymbol{u}^{i-1}$-$\boldsymbol{u}^{i}$ path.
For $x_q - x_r \ge c$, we want to show that $u^{i,k}_q - u^{i,k}_r \ge c$ holds.
By definition, we have $u^{i,k}_q - u^{i,k}_r = \max(u^{i-1,k}_q - 1, t_q) - \max(u^{i-1,k}_r - 1, t_r)$.
If $u^{i-1,k}_r - 1 \le t_r$,
then $\max(u^{i-1,k}_q - 1, t_q) - \max(u^{i-1,k}_r - 1, t_r) = \max(u^{i-1,k}_q - 1, t_q) - t_r \ge t_q - t_r \ge c$ holds.
If $u^{i-1,k}_r - 1 > t_r$,
then $\max(u^{i-1,k}_q - 1, t_q) - \max(u^{i-1,k}_r - 1, t_r) = \max(u^{i-1,k}_q - 1, t_q) - u^{i-1,k}_r - 1 \ge u^{i-1,k}_q - 1 - (u^{i-1,k}_r - 1) = u^{i-1,k}_q - u^{i-1,k}_r \ge c$ holds.
Therefore, we have $u^{i,k}_q - u^{i,k}_r \ge c$.
This completes the proof.
\end{proof}

\begin{proof}[Proof of Proposition~\ref{thm:UTVPI-P}]
From Lemma~\ref{lem:UTVPI-connectivity-equivalece},
it suffices to check if $\boldsymbol{u}^1$ is reachable from $\boldsymbol{u}^0$ in $G(I)$.
This can be done in polynomial time by using the pseudo-polynomial time algorithm in Proposition~\ref{cor:TVPI-pseudoP} for (unit) TVPI constraint with $D = \{0, 1\}$.
\end{proof}

We are now ready to show the main theorem of this subsection.

\begin{proof}[Proof of Theorem~\ref{thm:UILS-Z=1-P}]
We use Algorithm~\ref{alg:Z=1} but this time,
we use the polynomial time algorithms in Proposition~\ref{thm:UHorn-P} and \ref{thm:UTVPI-P}
to solve the problems in unit Horn and unit TVPI systems.
Then the theorem follows.
\end{proof}

Note that even for a unit monotone quadratic system (which is unit Horn and UTVPI at the same time),
the diameter of the solution graph is ${\rm \Theta}(dn)$ as Example~\ref{ex:Z=1-diameter-lower-bound} shows.
Therefore,
we obtain a polynomially solvable result for a reconfiguration problem in which the diameter of the solution graph can be exponential in the input size.
%Therefore, for a yes instance of the reconfiguration problem of a unit Horn and TVPI system,
%the length of a shortest path of the solution graph can be exponential in the input size.
%Therefore, our polynomial solvability result for unit ILS is might be the first result to show

\subsection{Bounded number of variables}

In this subsection, we show the following theorem.

\begin{theorem}
If the number of variables is a fixed constant,
then $\aP(1)$ reconfiguration is solvable in polynomial time.
\end{theorem}

\begin{proof}
To solve the reconfiguration problem of ILS(1),
		we determine if there exists a certificate for disconnectivity, which is used to show that the reconfiguration problem of ILS(1) is in coNP in Subsection~\ref{subsec:coNP-completeness}.
		
		Note that an instance is a no instance if and only if the exists a vector $\boldsymbol{w}$ satisfying the following conditions.\\
		$C_0$: $\boldsymbol{w} \in R$, \\
		$C_1$: $\boldsymbol{w} \le_{QH} \boldsymbol{s}$, \\
		$C_2$: $\boldsymbol{w}$ is locally $QH$-minimal, i.e.,
		%$\boldsymbol{w} - \alpha\boldsymbol{e_i} \not \in R$ for any $i=1, \dots, n$ and $\alpha > 0$, and \\
		$\boldsymbol{w} - \boldsymbol{e_j} \not \in R$ for any $j \in H$ and $\boldsymbol{w} - \sigma_j\boldsymbol{e_j} \not \in R$ for any $j \in Q \setminus \{ \ell \mid w_\ell=t_\ell \}$,
		where $\sigma_j = 1$ if $w_j > t_j$ and $-1$ if $w_j < t_j$, and \\
		$C_3$: there exists $j \in \{ 1, \dots, n \}$ such that $w_j >_{QH} t_j$.
		
		Note that $\boldsymbol{w} - \boldsymbol{e_j} \not \in R$ if and only if
		there exists an inequality $A_{i.} \boldsymbol{x} \ge b_i$ in $Ax \ge b$ such that
		$A_{i.} (\boldsymbol{w} - \boldsymbol{e_j}) < b_i$, i.e., $A_{i.} \boldsymbol{w} < A_{ij} + b_i$.
		
		To find a vector satisfying conditions $C_i$ ($0 \le i \le 3$),
		we guess for each $j$ the inequality that $\boldsymbol{w} - \boldsymbol{e_j}$ violates.
		We also guess which coordinate $j$ satisfies Condition $C_3$.
		We summarize our algorithm in Algorithm~\ref{alg:Z=1-bounded-number-of-variables_conf}.
		
		\begin{algorithm}
			\caption{Solving the reconfiguration problem of ILS(1) $I=(A,\boldsymbol{b},d)$ with fixed number of variables}
			\label{alg:Z=1-bounded-number-of-variables_conf}
			%	\begin{algorithmic}[1]
			%		\For{ $(i_1,\dots,i_n) \in \{1, \dots, m \}^n$ and $j \in \{ 1, \dots, n \}$}
			%		\For{a}{a}
			\For{$(i_1,\dots,i_n) \in \{1, \dots, m \}^n$ {\bf and} $j \in \{ 1, \dots, n \}$}{
				\If{there exists an integer vector $\boldsymbol{w}$ satisfying the following inequalities
					\begin{equation}\label{eq:certificate_conf}
					\left\{\begin{array}{ll}
					A\boldsymbol{w} \ge \boldsymbol{b}&\\
					\boldsymbol{w} \le_{QH} \boldsymbol{s}&\\
					A_{i_k.} \boldsymbol{w} < A_{i_k,k} + b_{i_k}& \text{{\rm (for all $k = 1, \dots, n$)}}\\
					w_j \ge_{QH} t_j + 1.&
					\end{array}
					\right.
					\end{equation}
				}{output ``NO'' and halt}
			}
			output ``YES'' and halt
			%	\end{algorithmic}
		\end{algorithm}
		
		Note that Eq.~\eqref{eq:certificate_conf} is a feasibility problem of linear inequalities with
		strict inequalities.
		However, we can replace $A_{i_k.} \boldsymbol{w} < A_{i_k,k} + b_{i_k}$ with $A_{i_k.} \boldsymbol{w} \le  A_{i_k,k} + b_{i_k} - \varepsilon$ for
		sufficiently small $\varepsilon >0$ without changing the feasibility.
		Therefore, we can check the feasibility of Eq.~\eqref{eq:certificate_conf} by solving the feasibility problem of an ILS.
		This can be done in polynomial time using the polynomial time algorithm for ILS with fixed number of variables~\cite{Len83}.
		Moreover, the number of loops in Algorithm~\ref{alg:Z=1-bounded-number-of-variables_conf} is
		polynomial in $m$ if the number of variable $n$ is a fixed constant.
		Therefore, the reconfiguration problem is solvable in polynomial time.
\end{proof}

\section{The case of $Z(I) < 1$}
\label{sec:Z<1}
%We show that in the case of $Z(I) < 1$, the reconfiguration problem of $I$ is always yes, implying that it can be solvable in constant time.

In this section, we consider the case of $Z(I)<1$ and show the following theorem.

\begin{theorem}\label{thm:Z<1}
	The reconfiguration problem of $\aP(\gamma)$ is always yes for any $\gamma < 1$.
\end{theorem}

From Lemma 3 in \cite{KiM16}, the input matrix $A$ admits an elimination ordering if $Z(I) < 1$.
Here, an \emph{elimination ordering} is a linear ordering of the columns of $A$ such that
$A$ becomes empty by repeatedly eliminating the columns $j$ that satisfies one of the following conditions.
\begin{description}
	\item[(i)]
	$a_{ij} > 0$ implies $a_{ij'} = 0$ with $j' \neq j$ for all $i=1, \dots , m$.
	\item[(ii)]
	$a_{ij} < 0$ implies $a_{ij'} = 0$ with $j' \neq j$ for all $i=1, \dots , m$.
\end{description}
From this property, we show the following useful lemma.

\begin{lemma}\label{lem:Z<1-sign-pattern}
	Let $I=(A,\boldsymbol{b},d)$ be an ILS with $Z(I) < 1$.
	Then the reconfiguration problem of $I$ can be reduced to
	that of $I'=(A',\boldsymbol{b}',d)$ such that
	$A'$ has the following property: 
	(P1) for any $i = 1, \dots, m$,  $A'_{ij} > 0$ implies that (a) $A'_{ij'} \leq 0$ for $j'<j$ and (b) $A'_{ij'}=0$ for $j'>j$.
\end{lemma}

\begin{proof}
	By rearranging the column vectors, we assume that $A$ admits an elimination ordering $(1, \dots , n)$.
	Let $A'$ be a matrix obtained from $A$
	%by replacing all the columns $A_{.j}$ eliminated by (ii) in the definition of elimination ordering with $-A_{.j}$.
	by replacing the $j$-th columns of $A$ with their opposite vectors for all $j$ eliminated by (ii) in the definition of elimination ordering.
	Then the reconfiguration problem of $I$ can be reduced to that of $I' = (A', \boldsymbol{b} - dA_{.j}, d)$ by Lemma~\ref{lem:polarity_change} in Section~\ref{sec:preliminaries}.
	Moreover, $A'$ admits an elimination ordering $(1, \dots , n)$ such that all columns are eliminated by (i) in the definition of elimination ordering.
	We now show that $A'$ satisfies the property (P1).
	%We show that, for each column $j$, $A'_{ij} > 0$ implies that (i) $A'_{ij'}\leq 0$ for $j'<j$ and (ii) $A'_{ij'}=0$ for $j'>j$.
	For (a), assume that $A'_{ij'} > 0$ holds for some $j'<j$.
	However, when $j'$ is eliminated, we have $A_{ij''} = 0$ for remaining $j''$s including $j$, which contradict $A_{ij} > 0$.
	Hence, (a) holds.
	(b) follows directly from the definition of elimination ordering.
	This completes the proof.
\end{proof}

To better appreciate the above lemma,
consider an ILS $I=(A,\boldsymbol{b},d)$ where all elements of $A$ are positive.
Then clearly $Z(I) = 0$ and, in fact, each column is eliminated by the condition (ii), which is vacantly holds for each column.
Moreover, let $A'$ be a matrix obtained from $A$ by replacing every column of $A$ with its opposite vector.
Then the elements of $A'$ are all negative, and (P1) vacantly holds for $A'$.

From Lemma~\ref{lem:Z<1-sign-pattern}, we assume that the input matrix $A$ has property (P1) in the following lemma.
In particular, $A$ is a Horn matrix since each row of it has at most one positive element.
Therefore, the feasible solutions of an ILS with input matrix $A$ has
a unique minimal solution from Lemma~\ref{lem:Horn-min-closed} in Subsection~\ref{subsec:basic}.
%From Lemma~\ref{lem:polarity_change}, we assume without loss of generality that the input matrix satisfies this property in this subsection.
Now, we show the following lemma.

\begin{lemma}\label{lem:Z<1-connectivity}
	Let $I$ be an instance of ILS which has at least one feasible solution, and with $Z(I) < 1$.
	Let $\boldsymbol{x}^*$ be a unique minimal solution to $I$ and $\boldsymbol{s}$ be any feasible solution to $I$.
	Then, there exists a path from $\boldsymbol{s}$ to $\boldsymbol{x}^*$ on $G(I)$.
	Consequently, $G(I)$ is a connected graph.
\end{lemma}

\begin{proof}
	Recall that $A$ is assumed to have the property (P1) in Lemma~\ref{lem:Z<1-sign-pattern}.
	Let $\boldsymbol{x}^* = (x^*_1, \dots, x^*_{n})$ be a unique minimal solution to $I$, and $\boldsymbol{s} = (s_1, \dots, s_{n})$ be any feasible solution to $I$.
	For each $k = 0, \dots, n$, we define $\boldsymbol{s}^k := (x^*_1, \dots, x^*_{k}, s_{k+1}, \dots, s_n)$, that is,
	$\boldsymbol{s}^k = (s^k_1, \dots, s^k_n)$ such that
	\begin{equation}
	s^k_i = \begin{cases}
	x^*_i & i \le k\\
	s_i & i > k.
	\end{cases}
	\end{equation}
		If we can show that $\boldsymbol{s}^k$ is a feasible solution to $I$ for each $k = 0, \dots, n$,
		then we immediately have that $\boldsymbol{s}$ and $\boldsymbol{x}^*$ are connected via a path
		$\boldsymbol{s} = \boldsymbol{s}^0 \rightarrow \boldsymbol{s}^1
		\rightarrow \dots \rightarrow \boldsymbol{s}^n = \boldsymbol{x}^*$.

		In the rest of this proof, we show the fact by induction on $k$.
		For $k=0$, $\boldsymbol{s}^0 (= \boldsymbol{s})$ is a feasible solution to $I$ by assumption.
		For $k>0$, assume that $\boldsymbol{s}^{k-1}$ is a feasible solution to $I$.
		%Consider an inequality $C$ of $I$ containing variable $x_{k}$.
		We then show that each inequality of $I$ is satisfied by $\boldsymbol{s}^{k}$.
		Consider the $i$-th inequality $\sum_{j=1}^n a_{ij}x_j \geq b_i$ of $I$.
		There are following three cases:
		\begin{description}
			\item{\textbf{Case 1:} $a_{ik} > 0$}\\
			In this case, $\sum_{j=1}^n a_{ij}x_j = \sum_{j=1}^k a_{ij}x_j$ holds from property (P1).
			%since $a_{ij} = 0$ for $j > k$ from the discussion before the lemma.
			Therefore, $\sum_{j=1}^n a_{ij}s^k_j = \sum_{j=1}^k a_{ij}s^k_j = \sum_{j=1}^k a_{ij}x^*_j \geq b_i$ holds.
			Thus,
			$\boldsymbol{s}^{k}$ satisfies the $i$-th inequality of $I$.
			\item{\textbf{Case 2:} $a_{ik} < 0$}\\
			In this case, $\sum_{j=1}^n a_{ij}s^k_j \geq \sum_{j=1}^n a_{ij}s^{k-1}_j$ holds.
			This is because $\boldsymbol{s}^{k}$ and $\boldsymbol{s}^{k-1}$ differ only on variable $x_k$,
			and $(a_{ik}s^k_k = )a_{ik}x^*_k \geq a_{ik}s_k (= a_{ik}s^{k-1}_k)$ holds
			since we have $x^*_k \leq s_k$ (by unique minimality of $\boldsymbol{x}^*$) and $a_{ik} < 0$.
			From the inductive assumption $\boldsymbol{s}^{k-1}$ is a feasible solution to $I$,
			which implies that $\sum_{j=1}^n a_{ij}s^k_j \geq \sum_{j=1}^n a_{ij}s^{k-1}_j \geq b_i$ holds.
			Therefore, $\boldsymbol{s}^{k}$ satisfies the $i$-th inequality of $I$.
			\item{\textbf{Case 3:} $a_{ik} = 0$}\\
			In this case, since $\boldsymbol{s}^{k}$ and $\boldsymbol{s}^{k-1}$ differ only on variable $x_k$,
			we have $\sum_{j=1}^n a_{ij}s^k_j = \sum_{j=1}^n a_{ij}s^{k-1}_j \geq b_i$.
		\end{description}
		Hence, $\boldsymbol{s}^{k}$ satisfies the $i$-th inequality of $I$ in all cases.
		Since $i$ is arbitrary, $\boldsymbol{s}^{k}$ is a feasible solution to $I$.
		This completes the proof.
\end{proof}

\if0
\begin{corollary}
	The reconfiguration problem of $\aP(1-\varepsilon)$ is trivial for any $\varepsilon > 0$.
\end{corollary}
\fi

Now we are ready to show Theorem~\ref{thm:Z<1}.

\begin{proof}[Proof of Theorem~\ref{thm:Z<1}.]
	From Lemma~\ref{lem:Z<1-connectivity},
	$G(I)$ is a connected graph.
	Therefore, the reconfiguration problem is always yes.
\end{proof}

\begin{corollary}\label{cor:Z<1-diameter}
	The diameter of $G(I)$ is ${\rm O}(n)$ if $Z(I) < 1$.
\end{corollary}

\begin{proof}
	From Lemma~\ref{lem:Z<1-connectivity},
	any two vertices of $G(I)$ are connected via the unique minimal solution $\boldsymbol{x}^*$ of $I$
	by a path of length at most $2n$.
	Therefore, the diameter of $G(I)$ is ${\rm O}(n)$.
	%$G(I)$ is a connected graph.
	%Therefore, the reconfiguration problem is always yes.
\end{proof}

We now show that the diameter of $G(I)$ can be ${\rm \Omega}(n)$ for ILS with index less than one.

\begin{example}\label{ex:Z<1-diameter-lower-bound}

Consider the following ILS with $n$ variables.

\begin{equation}\label{eq:Z<1-diameter_lower_bound}
\left\{
-x_j \ge -1 \ \ \  (j = 1, \dots, n)\\
\right.
\end{equation}
This ILS has index zero since $(\alpha_1, \alpha_2, \dots, \alpha_n,Z) = (1, 1, \dots, 1,0)$ is an optimal solution of LP \eqref{LP}.
Clearly, the set of solution of ILS~\eqref{eq:Z<1-diameter_lower_bound} is $\{0,1\}^n$.
Then the diameter of the solution graph of ILS~\eqref{eq:Z<1-diameter_lower_bound} is ${\rm \Omega}(n)$.
Indeed, consider a path from $(0, 0, \dots, 0)$ to $(1, 1, \dots, 1)$.
We can change a value of only one variable in each step.
Therefore,	the length of the path is at least $n$.
Thus, the diameter of the solution graph of ILS~\eqref{eq:Z<1-diameter_lower_bound} is ${\rm \Omega}(n)$.
\end{example}

\begin{theorem}\label{thm:Z<1-diameter-upper-and-lower-bound}
	The diameter of $G(I)$ is ${\rm \Theta}(n)$ if $Z(I) < 1$.
\end{theorem}
\begin{proof}
	This follows from corollary~\ref{cor:Z<1-diameter} and Example~\ref{ex:Z<1-diameter-lower-bound}.
\end{proof}

We finally note that the following example shows that
the solution graph might be disconnected if the index is one, which contrasts to Lemma~\ref{lem:Z<1-connectivity}.

\begin{example}\label{ex:Z=1-disconnected}
Consider the following ILS with $n=2$.

\begin{equation}\label{eq:Z=1-disconnected}
\left\{
\begin{array}{l}
x_1 - x_2 \ge 0 \\
-x_1 + x_2 \ge 0
\end{array}
\right.
\end{equation}

Then LP~\eqref{LP} is
\begin{eqnarray}\label{LP:Z=1-disconnected}
\begin{array}{cll}
\rm{minimize}     &Z\\
\rm{subject\ to}     &\alpha_1 + (1- \alpha_2) \leq Z &\\
&(1-\alpha_1) + \alpha_2 \leq Z  &\\
&0 \leq \alpha_1,\alpha_2 \leq 1. &
\end{array}
\end{eqnarray}
By summing up the first and second constraints in LP~\eqref{LP:Z=1-disconnected},
we obtain $2 \le 2Z$, namely $1 \le Z$.
On the other hand, $(\alpha_1,\alpha_2,Z) = (1/2,1/2,1)$ is a feasible solution to LP~\eqref{LP:Z=1-disconnected}.
Hence, the optimal value of LP~\eqref{LP:Z=1-disconnected} is one and the ILS has index one.

From ILS~\eqref{eq:Z=1-disconnected}, we have $x_1=x_2$ and thus the set of solutions of ILS~\eqref{eq:Z=1-disconnected} is $\{ (0,0), (1,1), \dots, (d,d) \}$.
Therefore, the solution graph is disconnected.
\end{example}

\section{The case of $Z(I) > 1$}
\label{sec:Z>1}
Recall that the reconfiguration problem of SAT is known to be PSPACE-complete~\cite{GKM09}.
Moreover, ILS is a generalization of SAT, that is, ILS can formulate SAT by representing each clause $(\bigvee_{j \in L^+} x_j \vee \bigvee_{j \in L^-} \overline{x}_j)$ as
$\sum_{j \in L^+} x_j + \sum_{j \in L^-} (1-x_j) \ge 1$ and setting $d=1$.
Therefore,
%since ILS is a generalization of SAT,
we can immediately have the PSPACE-completeness of the reconfiguration problem of ILS.
However, such an ILS may have a large complexity index $Z(I)$.
In this section, we show that such an index can be decreased to arbitrary small number greater than one.
Then we can say that the reconfiguration problem of ILS($\gamma$) is PSPACE-complete for any $\gamma > 1$.

To this end, we focus on the SAT problem for a while.
As ILS, let SAT($\gamma$) be the set of the SAT instances with index at most $\gamma$,
where index for SAT is computed via the above-mentioned ILS formulation, 
where we note that this index actually coincides with the complexity index for SAT introduced by Boros et al.~\cite{BCH94}.
We first explain the \emph{structural expressibility}, which is introduced in~\cite{GKM09} for reduction of connectivity problems.

Recall that relation $R$ is a subset of $D^k$ for some $k$.
Let $\Gamma$ be a finite set of logical relations, i.e., relations with $D=\{0,1\}$.
A CNF($\Gamma$)-formula over a set of variables $V = \{x_1, \dots, x_n\}$
is a finite conjunction $C_1 \wedge \dots \wedge C_n$ of clauses built using relations from $\Gamma$,
variables from $V$, and the constants $0$ and $1$.

\begin{definition}[\cite{GKM09}]\label{def:structural-expression}
A relation $R$ is \emph{structurally expressible} from a set of relations
$\Gamma$ if there exists a CNF($\Gamma$)-formula $\varphi$ such that the following conditions hold:
\begin{description}
\item[(1)] $R = \{\boldsymbol{a} \mid \exists \boldsymbol{y}\varphi(\boldsymbol{a}, \boldsymbol{y})\}$.
\item[(2)] For every $\boldsymbol{a} \in R$, the graph $G(\varphi(\boldsymbol{a}, \boldsymbol{y}))$ is connected.
\item[(3)] For $\boldsymbol{a}, \boldsymbol{b} \in R$ with $\dist(\boldsymbol{a},\boldsymbol{b}) = 1$, there exists $\boldsymbol{w}$ such that $(\boldsymbol{a},\boldsymbol{w})$ and $(\boldsymbol{b},\boldsymbol{w})$ are solutions of $\varphi$.
\end{description}
\end{definition}

In~\cite{GKM09}, the following useful lemma is shown.

\begin{lemma}[Corollary 3.3 in \cite{GKM09}]\label{lem:structural-reduction}
For two sets of relations $\Gamma$ and $\Gamma'$
assume that each relation $R \in \Gamma'$ is structurally expressible from $\Gamma$, and
that there exists a polynomial time algorithm that produces a structural expression from $\Gamma$ for each $R \in \Gamma'$.
Then there exists a polynomial-time reduction from the reconfiguration problem of $\Gamma'$ to that of $\Gamma$.
\end{lemma}

We are now ready to show our result.

\begin{theorem}\label{thm:SAT>1-PSPACE-c}
The reconfiguration problem of SAT($\gamma$) is PSPACE-complete for any $\gamma >1 $.
\end{theorem}
\begin{proof}
We reduce the reconfiguration problem of 3-SAT, which is known PSPACE-complete~\cite{GKM09}.
Let $\varphi$ be an instance of 3-SAT,
where $\varphi = \bigwedge_{i=1}^m C_i$ is a 3-CNF with $n$ variables and $m$ clauses, and $C_i = (\ell_{i_1} \vee \ell_{i_2} \vee \ell_{i_3})$ for $i = 1,\dots, m$.
Here, $1 \leq i_1,i_2,i_3 \leq n$ and $i_k$'s are distinct for $i = 1,\dots, m$, and
$\ell_{i_k}$ is a literal, i.e., unnegated or negated variable, for each $i_k$.

Let $\varepsilon = \gamma -1$ and $t =  \lceil 1/\varepsilon \rceil$.
For each clause $(\ell_{i_1} \vee \ell_{i_2} \vee \ell_{i_3})$,
we introduce $2(1+t)$ auxiliary variables $y_{i,0}, y_{i,1}, \dots, y_{i,t}$ and $z_{i,0}, z_{i,1}, \dots, z_{i,t}$.
Then we replace clause $(\ell_{i_1} \vee \ell_{i_2} \vee \ell_{i_3})$ with the clauses
$$\psi_i = (\ell_{i_1} \vee \overline{y}_{i,0})(y_{i,0} \vee \overline{y}_{i,1}) \cdots (y_{i,t-1} \vee \overline{y}_{i,t})
(y_{i,t} \vee \ell_{i_2} \vee z_{i,t})(\overline{z}_{i,t} \vee z_{i,t-1})\cdots (\overline{z}_{i,1}\vee z_{i,0})(\overline{z}_{i,0} \vee \ell_{i_3})$$
for each $i = 1, \dots, m$.
Let $\psi = \wedge_{i=1}^m \psi_i$.
Then we have the following claims.

\begin{claim}\label{claim1:in-Z>1-PSPACE-C}
$\psi$ is a structural expression of $\varphi$.
\end{claim}

\begin{proof}
We first see that the set of solutions of $\psi_i$ is as shown in Table~\ref{table:solution-psi_i} below.

\begin{table}[htb]
\caption{The set of solutions of $\psi_i$ for each assignment to $\ell_{i_1}, \ell_{i_2}$, and $\ell_{i_3}$}\label{table:solution-psi_i}
\centering
\begin{tabular}{cl}
\hline
$(\ell_{i_1}, \ell_{i_2}, \ell_{i_3})$ & $\boldsymbol{y},\boldsymbol{z}$\\ \hline
(0,0,0)& $\emptyset$ \\
(0,0,1)& $\boldsymbol{y}=\mathbf{0},\boldsymbol{z}=\mathbf{1}$\\
(0,1,0)& $\boldsymbol{y}=\mathbf{0},\boldsymbol{z}=\mathbf{0}$\\
(1,0,0)& $\boldsymbol{y}=\mathbf{1},\boldsymbol{z}=\mathbf{0}$\\
(0,1,1)& $\boldsymbol{y}=\mathbf{0}, z_{i,0} \ge z_{i,1} \ge \dots \ge z_{i,t}$\\
(1,0,1)& $\boldsymbol{y}=\mathbf{1}, z_{i,0} \ge z_{i,1} \ge \dots \ge z_{i,t}$ or $y_{i,0} \ge y_{i,1} \ge \dots \ge y_{i,t}, \boldsymbol{z}=\mathbf{1}$\\
(1,1,0)& $y_{i,0} \ge y_{i,1} \ge \dots \ge y_{i,t}, \boldsymbol{z}=\mathbf{0}$\\
(1,1,1)& $y_{i,0} \ge y_{i,1} \ge \dots \ge y_{i,t}, z_{i,0} \ge z_{i,1} \ge \dots \ge z_{i,t}$\\
\hline
\end{tabular}
\end{table}
\noindent For example, if we set $\ell_{i_1} = \ell_{i_2} = \ell_{i_3} = 0$, then we have $y_{i,0} = 0$ from clause $(\ell_{i_1} \vee \overline{y}_{i,0})$.
Note that we have $y_{i,0} \ge y_{i,1} \ge \dots \ge y_{i,t}$ and $z_{i,0} \ge z_{i,1} \ge \dots \ge z_{i,t}$ from the clauses in $\psi_i$ that contain no $\ell_{i_k}$'s.
Thus, $y_{i,t} = 0$ follows from $y_{i,0} = 0$.
Hence, we have $z_{i,t} = 1$ from $(y_{i,t} \vee \ell_{i_2} \vee z_{i,t})$ and $y_{i,t} = \ell_{i_2} = 0$.
On the other hand, $z_{i,0} = 0$ follows from $(\overline{z}_{i,0} \vee \ell_{i_3})$ and $\ell_{i_3} = 0$.
Hence, by $z_{i,0} \ge z_{i,1} \ge \dots \ge z_{i,t}$, we have $z_{i,t} = 0$.
This contradicts $z_{i,t} = 1$.
Therefore, $\ell_{i_1} = \ell_{i_2} = \ell_{i_3} = 0$ is not a part of any solution to $\psi_i$.
Similarly, we can check Table~\ref{table:solution-psi_i} holds for the other assignments to $\ell_{i_1}, \ell_{i_2}$, and $\ell_{i_3}$.

We now show the claim by checking (1), (2), and (3) in Definition~\ref{def:structural-expression} are satisfied from Table~\ref{table:solution-psi_i}.
Firstly, we have (1), since there exists no solution to $\ell_{i_1} = \ell_{i_2} = \ell_{i_3} = 0$ and there exists at least one solution to the other assignments to $\ell_{i_k}$'s.
For (2), we have to check if the solution space of $\boldsymbol{y}$ and $\boldsymbol{z}$ is connected for each assignment to $\ell_{i_1}, \ell_{i_2}$, and $\ell_{i_3}$.
If exactly one of $\ell_{i_k}$'s is 1, then $G(\psi_i(\ell_{i_1},\ell_{i_2},\ell_{i_3},\boldsymbol{y},\boldsymbol{z}))$ is connected since there is exactly one solution to $\boldsymbol{y}$ and $\boldsymbol{z}$.
For $(\ell_{i_1}, \ell_{i_2}, \ell_{i_3})=(0,1,1)$, we have the set
$$\{ (\mathbf{0};0, \dots, 0), (\mathbf{0};0, \dots, 0, 1), (\mathbf{0};0, \dots, 0, 1, 1), \dots, (\mathbf{0};1, \dots, 1) \}$$ of satisfying assignments to $(\boldsymbol{y},\boldsymbol{z})$ and this is clearly connected.
We can similarly show that $G(\psi_i(\ell_{i_1},\ell_{i_2},\ell_{i_3},\boldsymbol{y},\boldsymbol{z}))$ is connected for
$(\ell_{i_1}, \ell_{i_2}, \ell_{i_3})=(1,1,0)$.
For $(\ell_{i_1}, \ell_{i_2}, \ell_{i_3})=(1,0,1)$, we have the set
$$
\begin{array}{l}
\{ (\mathbf{0};0, \dots, 0), (\mathbf{0};0, \dots, 0, 1), (\mathbf{0};0, \dots, 0, 1, 1), \dots, (\mathbf{0};1, \dots, 1) = (0, \dots, 0;\mathbf{1}), \\
\ \ (0, \dots, 0, 1;\mathbf{1}), \dots, (1, \dots, 1;\mathbf{1}) \}
\end{array}
$$
of satisfying assignments to $(\boldsymbol{y},\boldsymbol{z})$ and this is again connected.
Finally, for $(\ell_{i_1}, \ell_{i_2}, \ell_{i_3})=(1,1,1)$, we have the set
$$
\begin{array}{l}
\{ (0, \dots, 0), (0, \dots, 0, 1), \dots, (1, \dots, 1)\} \times \{(0, \dots, 0), (0, \dots, 0, 1), \dots, (1, \dots, 1) \}
\end{array}
$$
of satisfying assignments to $(\boldsymbol{y},\boldsymbol{z})$, where $\times$ is the Cartesian product of two sets.
This is connected since the Cartesian product of two connected sets is connected.
Therefore, we have (2).
We note that this can be also proven by checking that $\psi_i$ admits an elimination ordering as an ILS and by Lemma~\ref{lem:Z<1-connectivity}.
For (3), observe that the set of satisfying assignment to $(\boldsymbol{y},\boldsymbol{z})$ is monotone,
i.e., if $\boldsymbol{a} \leq \boldsymbol{b}$ then we have $\{ (\boldsymbol{y},\boldsymbol{z}) \mid \psi_i(\boldsymbol{a},\boldsymbol{y},\boldsymbol{z})=1 \} \subseteq \{ (\boldsymbol{y},\boldsymbol{z}) \mid \psi_i(\boldsymbol{b},\boldsymbol{y},\boldsymbol{z})=1 \}$.
Therefore, for $\boldsymbol{a}, \boldsymbol{b} \in \{0,1\}^3\setminus \{(0,0,0)\}$ with $\dist(\boldsymbol{a},\boldsymbol{b}) = 1$, we can find a satisfying assignment $(\boldsymbol{y},\boldsymbol{z})$ for the smaller of $\boldsymbol{a}$ and $\boldsymbol{b}$ such that $\psi_i(\boldsymbol{a},\boldsymbol{y},\boldsymbol{z})=\psi_i(\boldsymbol{b},\boldsymbol{y},\boldsymbol{z})=1$ holds.
Therefore, we have shown (3).
This completes the proof.
\end{proof}

\begin{claim}
The complexity index $Z(\psi)$ is at most $1+ \varepsilon (=\gamma)$.
\end{claim}

\begin{proof}
For each $i=1, \dots, m$,
	consider the constraints of linear programming problem~\eqref{LP} corresponding to $\psi_i$:

	\begin{equation}\label{key}
	\begin{array}{ll}
	\alpha_{i_1} + 1 - \alpha_{y_{i,0}} &\leq Z\\
	\alpha_{y_{i,k}} + 1- \alpha_{y_{i,k+1}} &\leq Z \ (k=0, \dots, t-1)\\
	\alpha_{y_{i,t}} + \alpha_{i_2} + \alpha_{z_{i,t}} &\leq Z\\
	1-\alpha_{z_{i,k+1}} + \alpha_{z_{i,k}} &\leq Z \ (k=0, \dots, t-1)\\
	1-\alpha_{z_{i,0}} + \alpha_{i_3} &\leq Z.
	\end{array}
	\end{equation}

	Let $\alpha_{y_{i,k}} = \alpha_{z_{i,k}} = 1 - k/t$ for $k=0, 1, \dots, t$.
	Then, the inequalities above are feasible for $Z =1+\varepsilon$.
	Indeed, we have $\alpha_{i_1} + 1 - \alpha_{y_{i,0}} = \alpha_{i_1} \leq 1$.
	Moreover, $\alpha_{y_{i,k}} + 1 - \alpha_{y_{i,k+1}} = 1- \frac{k}{t} + \frac{k+1}{t} = 1+\frac{1}{t} \leq 1 + \varepsilon$ holds for $k=0, \dots, t-1$.
	We also have $\alpha_{y_{i,t}} + \alpha_{i_2} + \alpha_{z_{i,t}} = \alpha_{i_2} \leq 1$,
	$1-\alpha_{z_{i,k+1}} + \alpha_{z_{i,k}} = \frac{k+1}{t} + 1 - \frac{k}{t} = 1 + \frac{1}{t} \leq 1 + \varepsilon$ for $k=0, \dots, t-1$,
	and $1-\alpha_{z_{i,0}} + \alpha_{i_3} = \alpha_{i_3} \leq 1$.
	Therefore, all the inequalities are satisfied for $Z=1+\varepsilon$.
	Since linear programming problem~\eqref{LP} minimizes $Z$, we have $Z(\psi) \leq 1+\varepsilon = \gamma$.
\end{proof}

Now, the theorem follows from these claims and Lemma~\ref{lem:structural-reduction}.
\end{proof}

\begin{corollary}
	The reconfiguration problem of ILS($\gamma$) is PSPACE-complete for any $\gamma > 1$.
\end{corollary}

\begin{proof}
Since SAT is contained in ILS,
the statement follows form Theorem~\ref{thm:SAT>1-PSPACE-c}.
\end{proof}

Finally, we show that the diameter of $G(I)$ can be ${\rm \Omega}((d+2)\cdot 3^{\frac{n}{2}})$ even
for ILS $I$ with at most three variables per inequality for $d \geq 2$.

\begin{lemma}\label{lem:3ILS-diameter}
For $n$ even and $d \geq 2$, there exists an ILS $I_n$ with $n$ variables and $\frac{1}{2}n^2-\frac{1}{2}n+1$ inequalities
such that $G(I_n)$ is a path of length $\ell(n,d) = 2(d+2)\cdot 3^{\frac{n}{2}-1}-2$.
\end{lemma}

\begin{proof}
We first construct a path $P_n$ in $D^n$ of length $\ell(n,d)$ and
then provide an ILS $I_n$ such that $G(I_n)$ contains $P_n$ as a maximal connected component.

We inductively construct the path $P_n$.
For $n = 2$, let $V(P_2) = \{ (0,0), (0,1), (1,1), (1,2),$ $(2,2),\dots, (d-1,d), (d,d) \}$ which is a path of length $2d$.
For $n \geq 4$, assume that we have constructed path $P_{n-2}$ with two end points $\boldsymbol{s}_{n-2}$ and $\boldsymbol{t}_{n-2}$.
Then $V(P_n)$ is defined as follows.
For each $\boldsymbol{v} \in V(P_{n-2})$,
$V(P_n)$ contains three vertices $(\boldsymbol{v},0,0),(\boldsymbol{v},1,1)$, and $(\boldsymbol{v},2,2)$.
Note that the induced subgraph with only these vertices consists of three disjoint copies of $P_{n-2}$.
These components are connected by adding two vertices $(\boldsymbol{t}_{n-2},0,1)$ and $(\boldsymbol{s}_{n-2},1,2)$,
where the former connects $(\boldsymbol{t}_{n-2},0,0)$ and $(\boldsymbol{t}_{n-2},1,1)$, and the latter $(\boldsymbol{s}_{n-2},1,1)$ and $(\boldsymbol{s}_{n-2},2,2)$.
Now we have defined $V(P_n)$.
Observe that $P_n$ is a path from $\boldsymbol{s}_n = (\boldsymbol{s}_{n-2},0,0)$ to $\boldsymbol{t}_n = (\boldsymbol{t}_{n-2},2,2)$.
Indeed, it is a path $\boldsymbol{s}_n = (\boldsymbol{s}_{n-2},0,0) \rightarrow \dots \rightarrow (\boldsymbol{t}_{n-2},0,0) \rightarrow (\boldsymbol{t}_{n-2},0,1) \rightarrow (\boldsymbol{t}_{n-2},1,1) \rightarrow \dots \rightarrow (\boldsymbol{s}_{n-2},1,1) \rightarrow (\boldsymbol{s}_{n-2},1,2) \rightarrow (\boldsymbol{s}_{n-2},2,2) \rightarrow \dots \rightarrow (\boldsymbol{t}_{n-2},2,2) = \boldsymbol{t}_n$.
We now calculate its length $\ell(n,d)$.
From the above observation we have $\ell(n,d) = 3\cdot \ell(n-2,d) + 4$, and $\ell(2,d) = 2d$.
By solving this recursive equation, we obtain that $\ell(n,d) = 2(d+2)\cdot 3^{\frac{n}{2}-1} - 2$.
Moreover, we have $\boldsymbol{s}_2 = (0, 0), \boldsymbol{t}_2 = (d,d), \boldsymbol{s}_n =
(\boldsymbol{s}_{n-2}, 0, 0), \boldsymbol{t}_n = (\boldsymbol{t}_{n-2},2, 2)$ and hence
$\boldsymbol{s}_n = (0, \dots , 0), \boldsymbol{t}_n = (d,d, 2, \dots, 2)$.

We then construct an ILS $I_n$ such that $G(I_n)$ contains $P_n$ as a maximal connected component.
Let $I_2$ be defined as
\begin{equation}
\left\{
\begin{array}{l}
-x_1 + x_2 \ge 0\\
x_1 - x_2 \ge -1.
\end{array}
\right.
\end{equation}
Then the set of the feasible solutions of $I_2$ is indeed $\{ (0,0), (0,1), (1,1), (1,2), (2,2), \dots, (d-1,d), (d,d) \}$.
Assume we have $I_{n-2}$.
We add two variables $x_{n-1}$ and $x_n$, and the inequalities
\begin{eqnarray}
-x_{n-1} + x_n &\ge& 0 \nonumber\\
x_{n-1} - x_n &\ge& -1 \nonumber\\
x_j + 2dx_{n-1} - dx_n &\ge& 0 \ \ \ \ \ \, \text{for}\ \ \  j = 1, 2 \label{eq:(0,1)-j=1,2}\\
x_j + 4x_{n-1} - 2x_n &\ge& 0 \ \ \ \ \ \, \text{for}\ \ \  j = 3, \dots, n-2 \label{eq:(0,1)-j>=3}\\
-x_j + 2x_{n-1} - 4x_n &\ge& -6 \ \ \ \text{for}\ \ \  j = 1, \dots, n-2 \label{eq:(1,2)}
\end{eqnarray}
From the inequalities in \eqref{eq:(0,1)-j=1,2} and \eqref{eq:(0,1)-j>=3}, $(x_{n-1},x_n) = (0, 1)$ implies that $(x_1, \dots, x_{n-2}) = \boldsymbol{t}_{n-2} = (d,d, 2, \dots, 2)$.
Moreover, from the inequalities in \eqref{eq:(1,2)}, $(x_{n-1},x_n) = (1, 2)$ implies that $(x_1, \dots, x_{n-2}) = \boldsymbol{s}_{n-2} = (0, \dots , 0)$.
Furthermore, the inequalities in \eqref{eq:(1,2)} cannot be satisfied for $(x_{n-1},x_n) = (2, 3)$.
Therefore, $P_n$ cannot be prolonged in $G(I_n)$, implying that it is a maximal connected component of $G(I_n)$.
\end{proof}

\begin{theorem}\label{lem:Z>1-diameter}
	For infinitely many $n$, $d \geq 2$, and $\gamma>1$, there exists an ILS $I_n$ in ILS($\gamma$) with $n$ variables
	such that $G(I_n)$ has diameter
	$2(d+2)\cdot 3^{\sqrt{\frac{(\gamma -1) n}{8}}-1}-2$.
%$\ell(n,d) = 2(d+2)\cdot 3^{\sqrt{\frac{(\gamma -1) n}{8}}-1}-2$.
\end{theorem}

\begin{proof}
To show the theorem, we structurally express the instance $I_n$ in Lemma~\ref{lem:3ILS-diameter} by an instance $I'_N$ of ILS($\gamma$).
This expression is similar to that in Theorem~\ref{thm:SAT>1-PSPACE-c}.

Let $E_i: a_{i_1}x_{i_1} + a_{i_2}x_{i_2} + a_{i_3}x_{i_3} \ge b_i$ be the $i$-th inequality of $I_n$.
Let $\varepsilon = \gamma -1$ and $t =  \lceil 1/\varepsilon \rceil$.
We introduce $2(1+t)$ auxiliary variables $y_{i,0}, y_{i,1}, \dots, y_{i,t}$ and $z_{i,0}, z_{i,1}, \dots, z_{i,t}$ to
structurally express $E_i$.
Indeed, we structurally express $E_i$ by the set $\mathcal{E}_i$ of the inequalities
\begin{eqnarray}
a_{i_1}x_{i_1} - a_{i_1}y_{i,0} &\ge& 0 \label{eq:a_1x_1>=a_1y_0}\\
a_{i_1}y_{i,k} - a_{i_1}y_{i,k+1} &\ge& 0 \ \ \ \text{for}\ \ \ k = 0, \dots, t-1 \label{eq:a_1y_k>=a_1y_{k+1}}\\
a_{i_1}y_{i,t} + a_{i_2}x_{i_2} + a_{i_3}z_{i,t} &\ge& b \label{eq:a_1y_t+a_2x_2+a_3z_t>=b}\\
-a_{i_3}z_{i,k+1} + a_{i_3}z_{i,k} &\ge& 0 \ \ \ \text{for}\ \ \ k = 0, \dots, t-1 \label{eq:a_3z_k>=a_3z_{k+1}}\\
-a_{i_3}z_{i,0} + a_{i_3}x_{i_3} &\ge& 0, \label{eq:a_3x_3>=a_3z_0}
\end{eqnarray}
for each $i$.
Let $I_N' = \bigcup_i \mathcal{E}_i$.
We show that (i) $I_N'$ is a structural expression of $I_n$ and (ii) the complexity index $Z(I_N')$ is at most $1+ \varepsilon (=\gamma)$.

We first show that (i) holds by checking (1), (2), and (3) in Definition~\ref{def:structural-expression}.
For (1), assume that $(x_{i_1},x_{i_2},x_{i_3})$ satisfies $E_i$.
Then $\mathcal{E}_i$ is satisfied by setting $y_{i,0} = \dots y_{i,t} = x_{i_1}$ and $z_{i,0} = \dots = z_{i,t} = x_{i_3}$.
Conversely, if $\mathcal{E}_i$ is satisfied by $(x_{i_1},x_{i_2},x_{i_3},y_{i,0},\dots,y_{i,t},z_{i,0},\dots,z_{i,t})$, then
$(x_{i_1},x_{i_2},x_{i_3})$ satisfies $E_i$.
This is because we have $a_{i_1}x_{i_1} \ge a_{i_1}y_{i,0} \ge \dots \ge a_{i_1}y_{i,t}$ by the inequalities in \eqref{eq:a_1x_1>=a_1y_0} and \eqref{eq:a_1y_k>=a_1y_{k+1}} and $a_{i_3}x_{i_3} \ge a_{i_3}z_{i,0} \ge \dots \ge a_{i_3}z_{i,t}$ by the inequalities in \eqref{eq:a_3x_3>=a_3z_0} and \eqref{eq:a_3z_k>=a_3z_{k+1}}.
Thus, from $a_{i_1}y_{i,t} + a_{i_2}x_{i_2} + a_{i_3}z_{i,t} \ge b_i$ by the inequality in \eqref{eq:a_1y_t+a_2x_2+a_3z_t>=b},
we obtain $a_{i_1}x_{i_1} + a_{i_2}x_{i_2} + a_{i_3}x_{i_3} \ge b_i$.
Therefore, $(x_{i_1},x_{i_2},x_{i_3})$ satisfies $E_i$ and we have shown (1).
For (2), we observe that for each fixed $(x_{i_1},x_{i_2},x_{i_3})$, $\mathcal{E}_i$ admits an elimination ordering $(y_{i,0}, \dots, y_{i,t}, z_{i,t},\dots z_{i,0})$, where $y_{i,k}$ (resp., $z_{i,k}$) is eliminated by (ii)  (resp., (i)) in the definition of elimination ordering for each $k$.
Hence, by Lemma~\ref{lem:Z<1-connectivity}, the solution graph is connected and (2) is satisfied.
Finally, for (3), observe that the set of satisfying assignment to $(\boldsymbol{y},\boldsymbol{z})$ is monotone,
i.e., if $\boldsymbol{x} \leq \boldsymbol{x}'$ then we have $\{ (\boldsymbol{y},\boldsymbol{z}) \mid (\boldsymbol{x},\boldsymbol{y},\boldsymbol{z}) \text{ satisfies } \mathcal{E} \} \subseteq \{ (\boldsymbol{y},\boldsymbol{z}) \mid (\boldsymbol{x}',\boldsymbol{y},\boldsymbol{z}) \text{ satisfies } \mathcal{E} \}$.
Therefore, for satisfying assignments $\boldsymbol{x}, \boldsymbol{x}' \in D^3$ with $\dist(\boldsymbol{x},\boldsymbol{x}') = 1$, we can find a satisfying assignment $(\boldsymbol{y},\boldsymbol{z})$ for the smaller of $\boldsymbol{x}$ and $\boldsymbol{x}'$ such that both $(\boldsymbol{x},\boldsymbol{y},\boldsymbol{z})$ and $(\boldsymbol{x}',\boldsymbol{y},\boldsymbol{z})$ satisfy $\mathcal{E}$.
Therefore, we have shown (3).
This completes the proof of (i).

Similar to the proof in Theorem~\ref{thm:SAT>1-PSPACE-c}, we can show (ii).

Finally, we count the number $N$ of the variables in our structural expression $I_N'$ of $I_n$.
Note that $I_n$ has $\frac{1}{2}n^2-\frac{1}{2}n+1$ inequalities by Lemma~\ref{lem:3ILS-diameter}.
Since each inequality of $I_n$ is replaced by inequalities with $2(1+t)$ auxiliary variables,
the number $N$ of the variables in $I_N'$ is $N = n + 2(1+t)(\frac{1}{2}n^2-\frac{1}{2}n+1)$,
which is at most $\frac{2n^2}{\varepsilon}$.
Namely, we have $n \ge \sqrt{\frac{\varepsilon N}{2}}$.
Therefore, by substituting this to the diameter of $I_n$,
we obtain that the diameter of $I_N'$ is at least $2(d+2)\cdot 3^{\sqrt{\frac{\varepsilon N}{8}}-1}-2$.
\end{proof}

\section{Conclusion}
\label{sec:conclusion}
This paper investigates the complexity of the reconfiguration problem of ILS based on the complexity index introduced in~\cite{KiM16} and 
%In particular, we show that the problem is always yes if the index is less than one;
%weakly coNP-complete and pseudo-polynomially solvable if the index is exactly one; and
%PSPACE-complete if the index is greater than one.
obtains a complexity trichotomy.
On the way of showing this result, we also reveal the complexity of the reconfiguration problems of Horn and TVPI ILSes, 
ones of most studied subclasses of ILSes.
We also obtain a complexity dichotomy for the reconfiguration problem of unit integer linear systems and Boolean satisfiability problem in terms of the complexity index.

\section*{Acknowledgments}
The first author is partially supported by JSPS KAKENHI Grant Number JP17K12636, Japan.
The second author is partially supported by JST CREST Grant Number JPMJCR1402, 
and JSPS KAKENHI Grant Numbers JP17K12636 and JP18H04091, Japan.

%%
%% Bibliography
%%

%% Please use bibtex,

\bibliography{reconfiguration}

\end{document}